\documentclass[lettersize,journal]{IEEEtran}

\usepackage{amsmath,epsfig,amssymb,verbatim,amsopn,cite,multirow}
\usepackage{amsthm}
\usepackage{balance}
\usepackage{multirow}
\usepackage[usenames,dvipsnames]{color}
\usepackage[all]{xy}  
\usepackage{url}
\usepackage{amsfonts}
\usepackage{amssymb}
\usepackage{epsfig}
\usepackage{epstopdf}
\usepackage{bm}
\usepackage{balance}
\usepackage{graphicx}
\usepackage{subcaption}
\usepackage{footnote}
\usepackage{cancel}
\usepackage{algorithm,algorithmic}
\usepackage{dblfloatfix}

\newtheorem{Lemma}{Lemma}

\newtheorem{Remark}{Remark}

\newtheorem{proposition}{Proposition}

\newcounter{mytempeqcounter}
\usepackage{pifont}

\newcommand{\Ex}{\mathbb{E}}

\usepackage{tikz}
\newcommand{\tikzxmark}{%
\tikz[scale=0.15] {
    \draw[line width=0.6,line cap=round] (0,0) to [bend left=5] (1,1);
    \draw[line width=0.6,line cap=round] (0.1,0.85) to [bend right=3] (0.8,0.05);
}}

\newcommand{\qa}{{\bf a}}

\newcommand{\qe}{{\bf e}}

\newcommand{\qg}{{\bf g}}
\newcommand{\qh}{{\bf h}}

\newcommand{\qq}{{\bf q}}

\newcommand{\qt}{{\bf t}}

\newcommand{\qx}{{\bf x}}

\newcommand{\qA}{{\bf A}}
\newcommand{\qB}{{\bf B}}

\newcommand{\qD}{{\bf D}}

\newcommand{\qG}{{\bf G}}

\newcommand{\qI}{{\bf I}}

\newcommand{\qQ}{{\bf Q}}

\newcommand{\qU}{{\bf U}}

\newcommand{\qW}{{\bf W}}

\newcommand{\sen}{\mathrm{Sen}}
\newcommand{\DS}{\mathrm{DS}}
\newcommand{\DST}{\mathrm{DST}}

\newcommand{\Com}{\mathrm{Com}}
\newcommand{\ave}{\mathrm{ave}}
\newcommand{\IR}{\mathrm{IR}}

\newcommand{\tmlsen}{\qt_{ml}^{\mathrm{Sen}}}

\newcommand{\M}{\mathcal{M}}
\newcommand{\Ks}{\mathcal{S}_m}
\newcommand{\Kw}{\mathcal{W}_m}

\newcommand{\ZFcom}{\mathrm{ZF-Com}}
\newcommand{\MRcom}{\mathrm{MR-Com}}
\newcommand{\MR}{\text{MR}}
\newcommand{\SINR}{\mathrm{SINR}}
\newcommand{\SE}{\mathrm{SE}}
\newcommand{\tmkMRcom}{\qt_{mk}^{\mathrm{MR-Com}}}
\newcommand{\tmkZFcom}{\qt_{mk}^{\mathrm{ZF-Com}}}

\newcommand{\ETAC}{\boldsymbol{\eta}^{c}}
\newcommand{\ETAS}{\boldsymbol{\eta}^{s}}
\newcommand{\gammk}{\gamma_{mk}}
\newcommand{\betmk}{\beta_{mk}}
\newcommand{\gtmk}{\qg^{\dag}_{m k}}

\newcommand{\Bxi}{\boldsymbol{\xi}}

\newcommand{\Sm}{\mathcal{S}_{m}}
\newcommand{\Wm}{\mathcal{W}_{m}}
\newcommand{\Sn}{\sigma^2_n}

\newcommand{\BOmega}{\boldsymbol{\Omega}}
\newcommand{\etamls}{\eta^{s}_{ml}}
\newcommand{\etamlps}{\eta^{s}_{ml'}}
\newcommand{\etamkc}{{\eta}_{mk}^c}
\newcommand{\etamkpc}{{\eta}_{mk'}^c}

\newcommand{\KK}{\mathcal{K}}
\newcommand{\LL}{\mathcal{L}}
\newcommand{\MM}{\mathcal{M}}
\newcommand{\ZK}{\mathcal{Z}_k}
\newcommand{\MK}{\mathcal{M}_k}

\newcommand{\sumgammk}{\sum_{k\in\KK}\gammk}
\newcommand{\SEk}{\mathrm{SE}_{k}}

\newcommand{\trac}{\mathrm{tr}}
\newcommand{\zmzk}{\zeta_m^{\ZK}}
\newcommand{\zmmk}{\zeta_m^{\MK}}
\newcommand{\dkkw}{\delta_k^{\Kw}}
\newcommand{\dkks}{\delta_k^{\Ks}}

\title{Multiple-Target Detection in Cell-Free Massive MIMO-Assisted ISAC}
\author{Mohamed Elfiatoure, Mohammadali Mohammadi,~\IEEEmembership{Senior Member,~IEEE,}
Hien Quoc Ngo,~\IEEEmembership{Fellow,~IEEE,}\\
Hyundong Shin,~\IEEEmembership{Fellow,~IEEE,} and  Michail Matthaiou,~\IEEEmembership{Fellow,~IEEE}
\thanks{This work is a contribution by Project REASON, a UK Government funded project under the Future Open Networks Research Challenge (FONRC) sponsored by the Department of Science Innovation and Technology (DSIT). It was also supported by the U.K. Engineering and Physical Sciences Research
Council (EPSRC) (grant No. EP/X04047X/1). The work of  H. Q. Ngo
 was supported by the U.K. Research and Innovation Future
Leaders Fellowships under Grant MR/X010635/1, and a research grant from the Department for the Economy Northern Ireland under the US-Ireland R\&D Partnership Programme. The work of M. Mohammadi and M. Matthaiou was supported by the European
Research Council (ERC) under the European Union’s Horizon 2020 research
and innovation programme (grant agreement No. 101001331). \emph{(Corresponding authors: Michail~Matthaiou; Hyundong~Shin.)} }
\thanks{M. Elfiatoure, M. Mohammadi, H. Q. Ngo, and M. Matthaiou are with the Centre for Wireless Innovation (CWI), Queen's University Belfast, BT3 9DT Belfast, U.K., (email: \{melfiatoure01, m.mohammadi, hien.ngo, m.matthaiou\}@qub.ac.uk). M. Matthaiou is also with the Department of Electronic Engineering, Kyung Hee University, Yongin-si, Gyeonggi-do 17104, Republic of Korea.

H. Shin is with the Department of Electronics and Information Convergence Engineering, Kyung Hee University, 1732 Deogyeong-daero, Giheung-gu, Yongin-si, Gyeonggi-do 17104, Republic of Korea (e-mail: hshin@khu.ac.kr).

Parts of this paper appeared at the 2023 IEEE GLOBECOM conference~\cite{elfiatoure2023cell}.
}}

\begin{document}
\bstctlcite{IEEEexample:BSTcontrol}
\maketitle

\begin{abstract}
We propose a distributed implementation of integrated sensing and communication (ISAC) underpinned by a massive multiple input multiple output (CF-mMIMO) architecture without cells. Distributed multi-antenna access points (APs) simultaneously serve communication users (UEs) and emit probing signals towards multiple specified zones for sensing. 
The APs can switch between communication and sensing modes, and adjust their transmit power based on the network settings and sensing and communication operations' requirements. By considering local partial zero-forcing  and maximum-ratio-transmit precoding at the APs for communication and sensing, respectively, we first derive closed-form expressions for the spectral efficiency (SE) of the UEs and the mainlobe-to-average-sidelobe ratio (MASR) of the sensing zones. Then, a joint operation mode selection and power control design problem is formulated to maximize the SE fairness among the UEs, while ensuring specific levels of MASR for sensing zones. The complicated mixed-integer problem is relaxed and solved via a successive convex approximation approach. We further propose a low-complexity design, where the AP mode selection is designed through a greedy algorithm and then  power control is designed based on this chosen mode. Our findings reveal that the proposed scheme can consistently ensure a sensing success rate of $100\%$ for different network setups with a satisfactory fairness among all UEs. 

\end{abstract}


\vspace{-1em}
\section{Introduction}
The advent of integrated sensing and communication (ISAC) marks a significant leap forward for the sixth generation (6G) of wireless technology, introducing a pivotal shift in the use of limited spectral resources~\cite{liu2022integrated:dul}. This innovative approach aims at enhancing both the radar sensing capabilities and communication efficiency, by assigning  frequencies traditionally dedicated to sensing for wireless communication purposes as well. Such dual functionality supports high data-rate applications, contributing to the efficient utilization of the spectral resources. 
ISAC's innovative framework has catalyzed the development of two distinct system architectures: separated systems and co-located systems. A separated system, also referred to as coexistence communication and radar, relies on distinct devices to perform radar sensing and communication within the same frequency band, generating sensing and communication beams independently. This method has been thoroughly investigated in our previous studies~\cite{elfiatoure2023coexistence, Elfiatoure:JCIN:2023} among others. In contrast, co-located systems, or dual-functional radar-communication (DFRC), integrate these functionalities within a single device. This allows for the simultaneous detection of radar targets and communication with multiple users (UEs), facilitated by an ISAC base station (BS) that utilizes shared hardware components for concurrent beam generation~\cite{liu2020joint}. Recognized for its efficiency and cost-effectiveness, the DFRC architecture is heralded as a significant breakthrough for the beyond 5G (B5G) era, poised to redefine the network performance and optimization~\cite{liu2018dual,liu2022seventy,liu2018mu}.

Transitioning from separated to co-located systems reflects an intention to streamline the operational efficiency and reduce the complexity inherent in managing the exchange of side-information between radar and communication technologies~\cite{liu2022seventy,liu2018mu}. This evolution is notably supported by the integration of massive multiple-input multiple-output (mMIMO) technology, which particularly within co-located architectures, can effectively minimize the inter-user interference across the same time-frequency resources. Leveraging the substantial spatial degrees of freedom and intrinsic sensing capabilities, the fusion of mMIMO technology with ISAC is poised to deliver high-quality wireless communication, while ensuring high-resolution and robust sensing. Such advancements are instrumental in accelerating the development of various emerging applications, including, but not limited to, autonomous driving in intelligent transportation systems and unmanned aerial vehicle networks in smart cities~\cite{ali2020leveraging, zhang2020perceptive}.

Expanding on this groundwork, recent research endeavors, as highlighted in studies by~\cite{temiz2020dual} and~\cite{temiz2021optimized}, have delved into the intricacies of a dual-functional system that marries communication and radar functionalities within an mMIMO orthogonal frequency division mutliplexing (OFDM) architecture. These studies comprehensively address both downlink and uplink scenarios, offering insightful analytical derivations of the achievable rate and detection efficiency under various levels of channel state information (CSI) knowledge, encompassing both perfect and imperfect scenarios. Moreover, the work in~\cite{temiz2021dual} sought to optimize this dual-functional system. The focus was on maximizing the sum-rate and energy efficiency, while ensuring compliance with essential operational parameters, including a baseline target detection probability and the satisfaction of individual UE rate demands.

Another powerful network topology is cell-free mMIMO (CF-mMIMO), which offers great potential to significantly improve the network connectivity~\cite{Mohammadi:PROC.2024}. Diverging from traditional cellular setups that centralize all antennas at the BS, CF-mMIMO distributes antennas via the use of multiple access points (APs). These APs are then coordinated by several central processing units (CPUs), paving the way for a more distributed and flexible network architecture~\cite{ngo2024ultra,Matthaiou:COMMag:2021,Hien:cellfree,zhang2020prospective,Mohammadi:TWC:2024}. A CF-mMIMO ISAC topology has several advantages over the single-cell ISAC, including larger monitoring areas, broader sensing coverage, and a wider range of sensing angles. Nevertheless, in such systems, proper resource allocation plays a key role to facilitate both communication and sensing functionalities. 

\vspace{-1.5em}
\subsection{Related Works}
Some recent papers have focused on integrating ISAC capabilities within CF-mMIMO. These studies aim to harness the benefits of both ISAC and CF-mMIMO technologies to further enhance the network performance and service delivery, as evidenced by works such as those in~\cite{zeng2023integrated,Behdad:GC:2022,demirhan2023cell, da2023multi,Mao:TWC:2024}.
\begin{table*}\label{tabel:Survey}
	\centering
	\caption{\small Contrasting our contributions to the mMIMO ISAC literature}
	\vspace{-0.6em}
	\small
\begin{tabular}{|p{3.8cm}|p{1.7cm}|p{1.0cm}|p{1.4cm}|p{1.4cm}|p{1.4cm}|p{1.4cm}|}
	\hline
        \centering\textbf{Contributions} 
        &\centering \textbf{This paper}
        &\centering\cite{zeng2023integrated } 
        &\centering\cite{Behdad:GC:2022} 
        &\centering\cite{demirhan2023cell} 
        &\centering\cite{da2023multi}
        &\centering\cite{Mao:TWC:2024}
        \cr   
        \hline 
        
        CF-mMIMO    
         &\centering\checkmark  
        &\centering\checkmark  
        & \centering\checkmark  
        & \centering\checkmark     
        & \centering\checkmark
        & \centering\checkmark
        \cr
        
        \hline

        Power allocation     
        &\centering\checkmark  
        &\centering\checkmark  
        & \centering\checkmark  
        & \centering\checkmark     
        & \centering\tikzxmark
        & \centering\checkmark
         \cr
        
        \hline

        
       AP operation mode selection          
        &\centering\checkmark  
        &\centering\tikzxmark  
        &\centering\tikzxmark 
        &\centering\tikzxmark 
        & \centering\tikzxmark
        & \centering\tikzxmark
        \cr

        \hline
        Multiple targets         
        &\centering\checkmark  
        &\centering\tikzxmark  
        & \centering\tikzxmark
        & \centering\tikzxmark
       & \centering\tikzxmark
       & \centering\tikzxmark
        \cr
 
           \hline
        Asymptotic analysis         
        &\centering\checkmark  
        &\centering\tikzxmark   
        &\centering\tikzxmark
        &\centering\tikzxmark 
        &\centering\tikzxmark
        & \centering\tikzxmark
        \cr
        
        \hline

\end{tabular}
\vspace{-1.2em}
\label{Contribution}
\end{table*}
Among these contributions, Zeng \textit{et al.}\cite{zeng2023integrated}  made notable advancements in power allocation techniques within CF-mMIMO ISAC systems, optimizing the balance between communication and sensing functionalities, where the communication and sensing performances of the system are quantified using the sum communication rates and the Cram\'{e}r-Rao lower bound (CRLB), respectively. Behdad \textit{et al.}\cite{Behdad:GC:2022} further explored the operational dynamics of CF-mMIMO ISAC systems, demonstrating how transmit APs not only provide service to UEs but also partake in sensing operations for target location identification. Their research introduced a power distribution strategy that enhances the sensing signal-to-noise ratio , while adhering to signal-to-interference-plus-noise ratio (SINR) constraints for UEs.
Building on these themes, Demirhan \textit{et al.}~\cite{demirhan2023cell} tackled the challenges of beamforming design in CF-mMIMO ISAC systems. They proposed a joint beamforming approach designed to maximize the sensing SNR without compromising the communication SINR, highlighting the critical balance required for effective sensing and communication within these networks.
Da Silva~\textit{et al.}~\cite{da2023multi} examined the probability of an internal adversary inferring target location information from the received signal by considering the design of transmit precoders that jointly optimize sensing and communication requirements in a multi-static cell-free ISAC network.  Mao \textit{et al.}~\cite{Mao:TWC:2024} investigated downlink transmit beamforming design for the CF-mMIMO ISAC system and presented the achievable communication-sensing region of the ISAC system under finite resource constraints for three cases: 1) sensing-only, 2) communication-only, and 3) ISAC.

\vspace{-0.5em}
\subsection{Contributions}
 Contrary to the above studies, which assume static AP operation modes, our research introduces an innovative approach to CF-mMIMO ISAC networks by adopting \textit{dynamic AP operation mode selection}. This novel strategy aims to optimize the spectral efficiency (SE) fairness among the UEs, while also ensuring the system's capability to detect multiple designated targets. By leveraging long-term CSI, the APs are categorized into communication APs (C-APs) and sensing APs (S-APs) to simultaneously support both downlink communication and sensing tasks. Based on this AP mode selection, we propose an ISAC protocol where the S-APs primarily steer the beams to the target directions of interest acquired during the previous observation, in line with the foundational research on MIMO radar sensing outlined in~\cite{Stoica:TSP:2007,Liu:TSP:2020,Dong:TGCN:2023,Zhang:TWC:2024,Chen:TCOM:2024}. Meanwhile, the communication signals transmitted by the C-APs are exclusively used for downlink data transmission and are not involved in the sensing tasks. Our proposed ISAC protocol represents a distributed implementation of the colocated ISAC design presented in~\cite{Chen:TCOM:2024}. The authors in~\cite{Chen:TCOM:2024} investigated the emission of communication and sensing signals through different antenna groups from a single BS. The main objective is to enhance the independence of the sensing and communication functions. This approach is advantageous in communication-centric or sensing-centric ISAC designs, as it simplifies the achievement of optimal performance for the primary function.  Our protocol introduces an additional degree of freedom compared to~\cite{Chen:TCOM:2024} through flexible AP mode selection, offering two key benefits: 1) it mitigates the inter-system interference by positioning the C-APs (S-APs) closer to the UEs (sensing zones), thereby reducing the required transmit power, 2) it ensures seamless communication and sensing coverage with lower complexity and cost,\footnote{ The existing literature is mainly looking into single-cell scenarios, where a single MIMO BS serves one or multiple UEs while sensing one or multiple targets~\cite{Zhang:TWC:2024,Chen:TCOM:2024}. To overcome practical challenges, such as the high environmental dependence, limited coverage, and numerous blind spots, the use of reconfigurable intelligent surfaces has been proposed in~\cite{Zhang:TWC:2024,Chen:TCOM:2024}. However, these structures often introduce considerable complexity and overhead, particularly in the context of channel estimation and computational demands. Our design effectively addresses these challenges with lower complexity and cost.} and
3) this design reduces the potential risk of information leakage to suspicious targets, such as eavesdropping unmanned aerial vehicles. Since we only focus on the beampattern design for sensing, the processes of sensing parameter estimation and information demodulation are not considered in this paper. A detailed discussion can be found in~\cite{liu2020joint}.

 Compared to our recent work~\cite{elfiatoure2023cell}, where a single sensing zone area was considered and maximum-ratio transmission (MRT) was performed at all APs, we consider a multiple sensing zone scenario and apply local partial zero-forcing (PZF) and MRT precoding at the C-APs and S-APs, respectively. The principle behind the PZF design is that each AP only suppresses the inter-UE interference it causes to the strongest UEs, namely the UEs with the largest channel gain, while the inter-UE interference caused
to the weakest UEs is tolerated~\cite{Interdonato:TCOM:2020}. The main contributions of our paper can be summarized as follows:

\begin{itemize}

\item We develop a framework for analyzing the performance of a CF-mMIMO ISAC system with multiple C-APs and S-APs employing PZF and MRT precodings, respectively, under channel estimation errors. By leveraging the use-and-then-forget strategy, we derive closed-form expressions for the downlink SE of the communication UEs and mainlobe-to-average sensing ratio (MASR) of the sensing zones. We further pursue an asymptotic SE analysis, which discloses that when the number of C-APs and S-APs, denoted by $M_c$ and $M_s$ respectively, are large, we can scale down the transmit powers at the C-APs and S-APs proportionally to $1/M_c^2$ and $1/M_s^2$, respectively.

\item We formulate an interesting problem of joint AP operation mode selection and power control design, considering  per-AP power constraints and a MASR constraint for target detection in a multi-target environment. A new algorithm is developed to solve the challenging formulated mixed-integer non-convex problem. In particular, we transform the formulated problem into a more tractable problem with continuous variables only. Then, we solve the problem using successive convex approximation (SCA) techniques.

\item To achieve a performance-complexity tradeoff, we propose a greedy algorithm for AP operation mode selection. This algorithm iteratively selects the optimized mode for each AP by considering the constraints for sensing operation, while maximizing fairness among UEs. Then, we propose a power control design algorithm for fixed AP operation mode design.  

\item Our numerical results show that the proposed joint algorithm can provide noticeable fairness among the UEs, while ensuring successful sensing performance for all sensing zones. The greedy algorithm achieves an acceptable level of success in the sensing rate.

\end{itemize}

A comparison of our contributions against the state of the art in the space of CF-mMIMO ISAC is tabulated in Table~\ref{Contribution}.  A key distinguishing contribution of this work, compared to the existing literature, is the joint consideration of AP mode selection and power allocation, which enhances the SE of communication users while simultaneously fulfilling the network's sensing requirements. AP mode selection amplifies the flexibility in managing the interference between communication and sensing tasks, thereby improving the efficiency of resource allocation within the network. Moreover, this study addresses a multi-target scenario — a crucial aspect that has been overlooked in previous research. 

\subsection{Paper Organization and Notation}
The rest of this paper is organized as follows: In Section~\ref{sec:sys model}, we describe the system model for the proposed CF-mMIMO ISAC system with the corresponding communication and sensing metrics. In Section~\ref{sec:Asy}, we analyze the per-UE SE in the asymptotic regime. The proposed AP mode selection and power allocation schemes are discussed in Section~\ref{sec:power}. Finally, the numerical results and some discussions are provided in Section~\ref{sec:num}, followed by the conclusion remarks in Section~\ref{sec:conc}.

\textit{Notation:} We use bold lower (capital) case letters to denote vectors (matrices). The superscript $(\cdot)^\dag$ stands for the Hermitian operation;  $\mathbf{I}_N$ denotes the $N\times N$ identity matrix; $\trac(\cdot)$ and $(\cdot)^{-1}$ denote the trace operator and matrix inverse, respectively. A circular, zero-mean, symmetric complex Gaussian distribution having variance $\sigma^2$ is denoted by $\mathcal{CN}(0,\sigma^2)$. Finally, $\mathbb{E}\{\cdot\}$ denotes the statistical expectation.  

\section{System Model}~\label{sec:sys model}
We consider a CF-mMIMO ISAC system operating under time division duplex (TDD), where $M$ APs serve $K$ downlink UEs, while concurrently emitting probing signals towards $L$ specific sensing zones (Please see Fig.~\ref{fig:systemmodel}).  Each UE is equipped with a single antenna, and each AP equipped with a uniform linear array (ULA) of $N$ antennas. All APs and UEs function in a half-duplex mode. For the ease of exposition, we introduce the sets $\mathcal{M}\triangleq\{1,\ldots,M\}$, $\LL\triangleq \{1,\dots,L\}$, and $\KK\triangleq \{1,\dots,K\}$ to represent the indices of the APs, sensing areas, and UEs, respectively. Downlink communication and target detection are conducted concomitantly over the same frequency band. To accommodate the varying network demands, a dynamic AP operation mode selection strategy is implemented, determining the allocation of APs for downlink information transmission or radar sensing. A subset of APs, designated as communication-APs, termed as C-APs, is responsible for delivering information to the UEs. In contrast, the rest, labeled as S-APs, are employed for target detection. This categorization enables a specialized functionality of each AP group, where the C-APs aim to optimize the communication quality, while the S-APs enhance the radar sensing accuracy.
\vspace{-0.7em}
\subsection{Channel Model and Uplink Training}
We assume a quasi-static channel model, with each channel coherence interval spanning a duration of $\tau$ symbols. The duration of the training is denoted as $\tau_t$, while the duration of downlink information transfer and target detection is $(\tau-\tau_t)$.

For the sensing channel model, we assume there is a line-of-sight (LoS) path between the sensing area and each AP, which is a commonly adopted model in the literature~\cite{Liu:TSP:2020,Behdad:GC:2022,demirhan2023cell}. The direction of sensing zone $l\in\mathcal{L}$ is denoted by  $\theta_{t,ml}$ which is angle of departure (AoD) associated with AP $m$ towards zone $\ell$. The LoS channel between AP $m$ and sensing area $l$ is given by 
\begin{align}~\label{eq:barhkbarH2}
    \bar{\qg}_{ml} &=\sqrt{\zeta_{ml}} \qa_N(\theta_{t,ml}),~\forall m\in \M,
\end{align}
where $\zeta_{ml}$ denotes the path loss from AP $m$ to sensing zone $l$~\cite{Mao:TWC:2024,Behdad:GC:2022}. Consistent with previous literature, we neglect the paths resulting from multi-reflections off other objects caused by the presence of the target~\cite{Zhao:CLET:2022,demirhan2023cell,Behdad:GC:2022}. Moreover, $\qa_N(\theta_{t,ml})\in\mathbb{C}^{N\times 1}$ denotes the antenna array response vector from AP $m$ towards the sensing area $l$, whose  $n$-th entry is given by
\vspace{-0.5em}
\begin{align}~\label{eq:LoSarray}
    &[\qa_N(\theta_{t,ml})]_{n}
    = \exp\bigg( j\frac{2\pi d}{\lambda}  (n-1) \sin (\theta_{t,ml}) \bigg),
    \end{align}
where $d$ and $\lambda$ denote the AP antenna spacing and carrier wavelength, respectively.

The channel vector between the $m$-th AP and $k$-th UE is modeled as $\qg_{mk}= \sqrt{\beta_{mk}} \qh_{mk}$, where $\beta_{mk} $ is the large scale fading coefficient, while $\qh_{mk} \in \mathbb{C}^{N \times 1}$ is the small-scale fading vector, whose elements are independent and identically distributed $\mathcal{CN} (0, 1)$ random variables (RVs). 

An uplink training process is implemented to acquire the local CSI between each AP and all UEs.
In each coherence block of length $\tau$, all UEs are assumed to transmit their pairwisely orthogonal pilot sequence of length $\tau_t$ to all APs, which requires $\tau_t\geq K$. At AP $m$, $\qg_{mk}$ is estimated by using the received pilot signals and applying the minimum mean-square error (MMSE) estimation technique. By following~\cite{Hien:cellfree}, the MMSE estimate  $\hat{\qg}_{mk}$ of $\qg_{mk}$ is obtained as $\hat{\qg}_{mk}  \sim \mathcal{CN}\left(\boldsymbol{0}, \gammk \qI_N \right)$, where 
$\gammk \triangleq\frac { \tau_{t} \rho_{t}\beta _{mk}^{2}} 
{ \tau_{t} \rho_{t}\beta _{mk}+1}$, while $\rho_{t}$ represents the normalized transmit power of each pilot symbol.
\begin{figure}[t]
\centering
\includegraphics[width=0.47\textwidth]{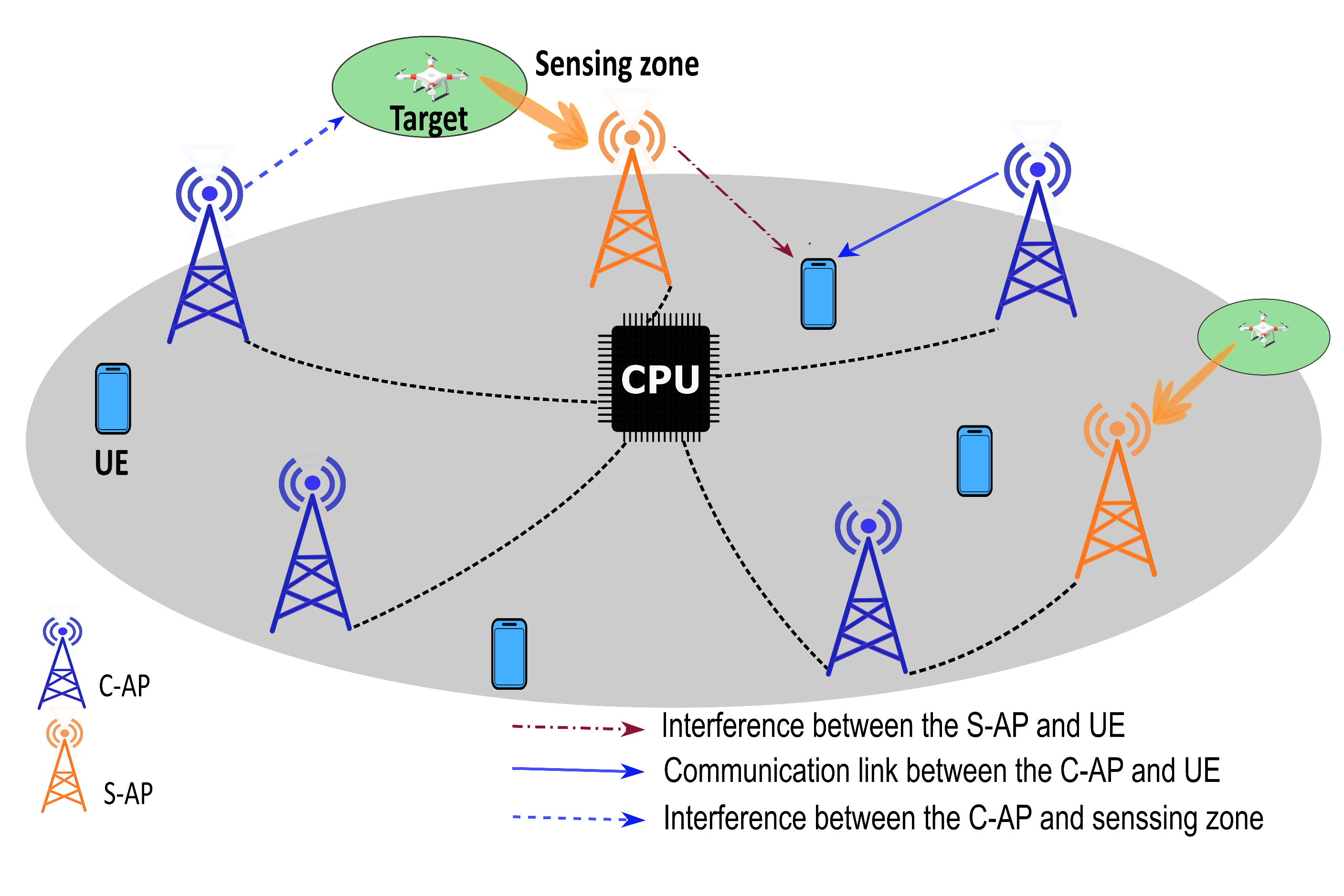}
\vspace{ 0em}
\caption{\small Illustration of the CF-mMIMO ISAC system.}
\vspace{1em}
\label{fig:systemmodel}
\end{figure}
\vspace{-0.9em}
\subsection{Probing and Data Signal Transmission}
AP operation mode selection is performed by considering large-scale fading effects and relying on the statistical CSI. The binary variables used to indicate the operation mode for each AP $m$ are defined as
\begin{align}
a_{m} =
 \begin{cases} &\text{1, if AP $m$ operates as C-AP}\\  
 &\text{0, if AP $m$ operates as S-AP.}
 \end{cases}
\end{align}

In the distributed ISAC protocol, the communication signals are not utilized for sensing. Instead, dedicated sensing signals are transmitted from the S-APs to achieve high beamforming gain for sensing, thereby minimizing the interference at the communication users. Specifically, by assigning APs near the target zones for sensing tasks and those near UEs for downlink data transmission, the transmit power of the APs can be reduced, thereby effectively managing interference. Additionally, this distributed approach simplifies the beamforming and power control, since each AP is responsible for a single task. Moreover, this design reduces the potential risk of information leakage to suspicious targets, such as eavesdropping unmanned aerial vehicles. Let $\mathbf{x}_{c,m}$ and $\mathbf{x}_{r,m}$ denote the data and probing signals, respectively. The signal vector transmitted from AP $m$ can be expressed as 
\begin{align}
\mathbf{x}_m=\sqrt{a_{m}} {\qx}_{c,m}+\sqrt{(1-a_{m})}{\qx_{r,m}}.
\end{align}

Under the assumption that the propagation path between the APs and targets is LoS, the
baseband signal at direction, $\theta_{t,ml}$ can be written as\footnote{ An alternative beamforming approach, e.g., ZF, could be considered at the S-APs to mitigate the interference between different targets. This design requires a new mathematical framework and, therefore, remains a future research direction.}
\begin{align}
   \qx_{r,m}=\sum\nolimits _{l\in\LL}\sqrt{\rho\etamls } \qa_N(\theta_{t,ml}) x_{r,ml},
\end{align}
where $\rho$ denotes the maximum normalized transmit power; $\etamls$ is the power control coefficient at S-AP $m$ related to sensing zone $l$, while $ x_{r,ml}$, with $\Ex\{\vert  x_{r,ml}\vert^2\}=1$, is the radar sensing symbol from S-AP $m$ for sensing the $l$th zone. This directed beampattern allows each S-AP to transmit a beam focused on a specific sector, facilitating spatially distributed sensing across the network.
 Therefore, to detect the presence of a target within its designated sector, each S-AP tailors its beampattern to concentrate the radiated energy towards a specific direction. This focused beamforming ensures that the radar echo from the target is sufficiently strong for reliable detection.
 
\subsection{Downlink Communication and SE}
In this section, we investigate the performance of the PZF precoding scheme. The PZF design effectively overcomes the interference management limitations of MRT (which arise from inter-UE interference) and the implementation challenges of ZF (i.e., the number of antennas at each AP must satisfy the condition $N>K$) within a distributed framework.\footnote{ The PZF precoding design facilitates fully distributed interference suppression in a scalable manner, minimizing the need for coordination among the APs and thereby avoiding unmanageable fronthaul traffic. However, this limited coordination restricts the direct utilization of signals transmitted by the C-APs for sensing tasks within the network, unless additional inter-AP coordination mechanisms are implemented.}

PZF precoding emerges as an innovative method with the potential to combine the benefits of both ZF and MR approaches. The cornerstone of PZF lies in its ability to suppress interference caused to the strongest UEs, i.e., UEs with the highest channel gain, while tolerating the interference caused to the weakest UEs.
In more detail, for any given AP (denoted as AP $m$), the set of active UEs is conceptually divided into two distinct subsets: (i) strong UEs, and (ii) weak UEs. The sets of indices for strong and weak UEs are represented as $\Sm\subset \{1,\ldots,K\}$, and $\Wm\subset \{1,\ldots,K\}$ respectively. It should be noted that $\Sm \cap \Wm=\phi$ and $|\Sm| + |\Wm|=K$.
The grouping strategy can be determined using various criteria. For example, it may be based on the mean-square of the channel gain: a UE $k$ is allocated to $\Sm$ if $\beta_{mk}$ exceeds a predetermined threshold, otherwise, $k$ is placed in $\Wm$.
PZF operates locally in the following manner: AP $m$ transmits to all the UEs $k \in \Sm$ using ZF, and to all the UEs $k \in \Wm$ using MR.  The  UE grouping can be based on different criteria. Inspired by~\cite{Ngo:TGCN:2018,Interdonato:TCOM:2020}, the UE grouping strategy in PZF relies on the following rule
\vspace{-0.2em}
\begin{align}~\label{eq:groupin:criterion}
   \sum\nolimits_{k=1}^{\vert\Sm\vert} \frac{\bar{\beta}_{mk}}{\sum\nolimits_{t=1}^{K} \beta_{mt}} \geq \varrho\%,
\end{align}
where AP $m$ constructs its set $\Sm$ by selecting the UEs that contribute to, at least, $\varrho\%$ of the overall channel
gain. Note that in~\eqref{eq:groupin:criterion}, $\{\bar{\beta}_{m1},\ldots,\bar{\beta}_{mK}\}$ indicates the set of the large-scale fading coefficients sorted in descending order.

Thus, the transmit signal from AP $m$ can be expressed as
\begin{align}~\label{eq:xcm}
   {\qx}_{c,m}=&\sum\nolimits_{k\in\Ks} \sqrt{\rho\etamkc}   \tmkZFcom x_{c,k} 
    \nonumber\\
    &\!+\!\sum\nolimits_{k\in\Kw} \sqrt{\rho\etamkc}  \tmkMRcom x_{c,k},
\end{align}
where $\etamkc$ denotes the power control coefficient between C-AP $m$ and UE $k$, $k\in\KK$; $x_{c,k}$ is the communication symbol satisfying  $\Ex\{\vert x_{c,k}\vert^2\}=1$; $ \qt^{\ZFcom}_{mk}$ and $ \qt^{\MRcom}_{mk}$ represent the ZF and MR precoding vectors, respectively, given by
\begin{align}\label{eq:Tzf}
\qt^{\ZFcom}_{mk} &= \gammk \hat{\qG}_{\Sm} \left (\hat{\qG}_{\Sm} ^{\dag} \hat{\qG}_{\Sm}\right) ^{-1} \qe_k,\\
\qt^{\MRcom}_{mk}&=\hat{\qg}_{mk},
\label{eq:Tmr}
\end{align}
where $\qe_k$ is the $k$th column of  $\qI_{\vert \Sm \vert}$ and $\mathbf{\hat{G}}_{\Sm}=[\mathbf{\hat{g}}^{\dag}_{m k}: k \in \Sm]$. Therefore, for any pair of  UEs $k$ and $k'$ belonging to the $\Sm$, we have
\begin{equation}
    \mathbf{\hat{g}}^{\dag}_{m k'}\qt^{\ZFcom}_{mk}=
    \begin{cases}
      \gammk & \text{if } k = k',\\
      0 & \text{otherwise. }
    \end{cases}
\end{equation}

The power control coefficients at AP $m$ are chosen to satisfy the power constraint at each S-AP and C-AP, respectively, i.e.,
\begin{align}~\label{eq:powerconst}
{a_{m}} \Ex\{\|\mathbf{x}_{c,m}\|^2\}+({1-a_{m}})\Ex\{\|\mathbf{x}_{r,m}\|^2\}\leq \rho.
\end{align}

For the sake of notation simplicity, we introduce $\ZK$ and $\MK$ as the set of indices of APs that transmit to UE $k$ using $\qt^{\ZFcom}_{mk}$  and $ \qt^{\MRcom}_{mk}$ respectively, given as
\begin{subequations}
 \begin{align}
\ZK&\triangleq \big\lbrace m:k \in \Sm,m=1,\ldots,M \big\rbrace,\\ 
\MK&\triangleq \big\lbrace m:k \in \Wm,m=1,\ldots,M \big\rbrace, 
\end{align}   
\end{subequations}
with $ \ZK \cap  \MK = \phi$ and  $\ZK \cup \MK = \mathcal{M}$.
The PZF precoding approach provides a sophisticated balance between ZF and MR, catering to interference suppression for the strongest UEs and signal strength maximization for the weakest ones.

The received signal at the $k$-th UE can be represented as
\vspace{0.0em}
\begin{align}\label{eq:yk}
{y}_k&=\Big( \sum\nolimits_{m\in{\ZK}}  \sqrt{a_m\rho \etamkc} \gtmk \qt^{\ZFcom}_{mk}\\
&\hspace{0em}+\sum\nolimits_{m\in{\MK}}  \sqrt{a_m\rho \etamkc} \gtmk \qt^{\MRcom}_{mk}   \Big)x_{c,k}\nonumber\\
&\hspace{0em}+\sum\nolimits_{k'\in\KK \setminus k}  \Big( \sum\nolimits_{m\in{\ZK}} \sqrt{a_m \rho \etamkpc} \gtmk \qt^{\ZFcom}_{mk'} \nonumber\\
&\hspace{0em}
+\sum\nolimits_{m\in\MK}  \sqrt{a_m\rho \etamkpc} \gtmk \qt^{\MRcom}_{mk'}   \Big)x_{c,k'}\nonumber\\
&\hspace{0em}+\!{\sum\nolimits_{m\in\mathcal{M}}\! \sum\nolimits_{l\in\LL} \!\sqrt{\!\left(1\!-\!a_m\right)\rho \etamls}\gtmk}   \qa_N(\theta_{t,ml}) x_{r,ml}\!+\!{n}_k, \nonumber
\end{align}
where the second term is the inter-user interference, the third term represents the interference from S-APs, while $n_k\sim\mathcal{CN}(0,1)$ denotes the additive white Gaussian noise at the $k$-th UE.

We assume that each UE has knowledge of the channel statistics but not of the instantaneous channels. Therefore, to derive the downlink SE at UE $k$, we utilize the so-called \textit{use-and-then-forget} bounding technique, that is widely used  in the literature for cell-free massive MIMO systems~\cite{Hien:cellfree,emil20TWC,Mohammadi:JSAC:2023}. Therefore, we rewrite~\eqref{eq:yk} as
\begin{align}\label{eq:yk:utf}
    y_k &=\!  \mathrm{DS}_k  x_{c,k} \!+\!
    \mathrm{BU}_k x_{c,k}
     \!+\!\sum\nolimits_{k'\in\KK \setminus k}
     \mathrm{IUI}_{kk'}
     x_{c,k'}
    \!+\! 
    \IR_{k}  \!+\! n_k,
\end{align}
where 
\vspace{-0.2em}
\begin{subequations}
  \begin{align}
 \mathrm{DS}_k  &\triangleq 
 \Ex\Big\{ \sum\nolimits_{m\in\ZK} \sqrt{a_m \rho \etamkc} \gtmk \qt^{\ZFcom}_{mk} \nonumber\\
 &\hspace{1em}
 +\sum\nolimits_{m\in\MK}  \sqrt{a_m\rho \etamkc} \gtmk \qt^{\MRcom}_{mk}  \Big\},
 \\
 \mathrm{BU}_k  &\triangleq  
  \Big( \sum\nolimits_{m\in\ZK} \sqrt{ a_m\rho \etamkpc} \gtmk \qt^{\ZFcom}_{mk'} \nonumber\\
  &
  \hspace{-1em}
  +\sum\nolimits_{m\in\MK}  \sqrt{a_m\rho \etamkpc} \gtmk \qt^{\MRcom}_{mk'}    \Big)- \mathrm{DS}_k,
 \\
 \mathrm{IUI}_{kk'} &\triangleq 
   \sum\nolimits_{m\in\ZK} \sqrt{ a_m\rho \etamkpc} \gtmk \qt^{\ZFcom}_{mk'} \nonumber\\
 &
 \hspace{1em}
 +\sum\nolimits_{m\in\MK}  \sqrt{a_m\rho \etamkpc} \gtmk \qt^{\MRcom}_{mk'}, 
   \\
  \IR_{k} &\!\triangleq\! \sum\nolimits\nolimits_{m\in\M} \!\sum\nolimits_{l\in\LL} \!\!
  \sqrt{(\!1\!-\!a_m)\rho\etamls }\qg_{mk}^T  \qa_N(\theta_{t,ml}) x_{r,ml},~\label{eq:IRk}  
 \end{align}  
\end{subequations}
represent the strength of the desired signal ($\mathrm{DS}_k$), 
the beamforming gain uncertainty ($\mathrm{BU}_k$), the interference caused by the $k'$-th UE ($\mathrm{IUI}_{kk'}$) and the interference caused by S-APs ($\IR_k$), respectively. By invoking~\eqref{eq:yk:utf}, an achievable downlink SE at the $k$-th UE can be expressed as
\vspace{0.1em}
\begin{align}~\label{eq:SE}
 \SEk=\Big(1-\frac{\tau_t}{\tau}\Big)\log _{2}\big(1 + \SINR_k\big),  
\end{align}
where 
$\SINR_k\!=$ 
\begin{align}~\label{eq:dLSE}
&\frac{
                 \big\vert  \mathrm{DS}_k  \big\vert^2
                 }
                 {  
                 \Ex\big\{ \!\big\vert  \mathrm{BU}_k  \big\vert^2\!\big\} \!+\!
                 \sum_{k'\in\KK \setminus k}
                  \!\Ex\big\{ \!\big\vert \mathrm{IUI}_{kk'} \big\vert^2\!\big\}
                  \! +\!  \Ex\big\{ \!\big\vert \IR_k \big\vert^2\!\big\}  \!+\!  1}\!.
\end{align}

For the simplicity of notation, let us introduce a pair of binary variables to indicate the group assignment for each downlink UE $k$ and AP $m$ in the PZF combining scheme as
\vspace{-0.1em}
\begin{align}
    \zeta_m^{\ZK}&=
    \begin{cases}
     1 & \text{if } m \in\ZK,\\
      0 & \text{otherwise}.
    \end{cases}
   \quad
    \zeta_m^{\MK}&=
    \begin{cases}
     1 & \text{if } m \in \MK,\\
      0 & \text{otherwise}.
    \end{cases}
\end{align}    
\begin{proposition}~\label{Prop:SE:PPZF}
 The SE achieved by the PZF scheme is represented by~\eqref{eq:SE}, where $\SINR_k$ is given in~\eqref{eq:dLSNIRf} at the top of the next page.    
\end{proposition}
\begin{proof}
    See Appendix~\ref{Prop:SE:PPZF:proof}.
\end{proof}

\begin{Remark}
Setting $\zeta_m^{\MK}=1$ and $\zeta_m^{\ZK}=0$, $\forall m \in\MM$, makes the SINR given by~\eqref{eq:dLSNIRf} be reduced to that achieved by the MR precoding scheme. Similarly, if $\zeta_m^{\MK}=0$ and $\zeta_m^{\ZK}=1$, $\forall m \in\MM$, the SINR in~\eqref{eq:dLSNIRf} reduces to that achieved by the ZF precoding scheme, provided that $N>K$.
\end{Remark}

\begin{figure*}
  \begin{align}~\label{eq:dLSNIRf}
&\SINR_k(\mathbf{a},  \boldsymbol{\eta}^{c}, \boldsymbol{\eta}^{s})=\frac{
\rho\Big(\!\sum_{m\in\MM}  \sqrt{ a_m  \etamkc }\gammk (\zmzk + N\zmmk)\Big)^2
                 }
                 {                  
                \rho\sum_{m\in\MM}\sum_{k'\in\KK }
                \!\!\Big(
         a_m \etamkpc\gamma_{mk'} \big(\zmzk 
        \frac{ \betmk-\gammk}{N-\vert\Sm\vert} +\zmmk N\betmk\big) \Big) \!+\!  \rho\sum\nolimits_{m\in\M}\!\sum\nolimits _{l\in\LL} \!\etamls N (1-a_m) \betmk \!+\!  1
        },        
\end{align} 
\hrulefill
\vspace{-5mm}
\end{figure*}
\vspace{-1.5em}
\subsection{Sensing Operation and MASR}
Our goal for the sensing operation is to steer the beam to the target directions of interest to facilitate high-quality capturing of reflected/echo signals, thereby extracting richer sensing information. To achieve this, we focus on efficient transmit spatial beampattern design to concentrate the signal power while maintaining a low sidelobe level~\cite{Stoica:TSP:2007,Liu:TSP:2018,Liu:TSP:2020}. Consequently, multistatic sensing can be performed using our waveform design at the S-APs. The development of a sensing scheme is beyond the scope of this work and will be addressed in future research.

To facilitate our sensing design, for any given direction $\theta_{t,m\in\mathcal{M}}$, we define the sensing beampattern gain as the transmit signal power distribution at that direction~\cite{Stoica:TSP:2007,Liu:TSP:2018,Hua:TWC:2024}:
\vspace{-0.5em}
\begin{align}~\label{eq:pavet1}
P^{\ave}(\theta_{t,m\in\M})
\!=&
\Ex\bigg\{
\Big\vert
\sum\nolimits_{m\in\M} \mathbf{a}_N^\dag(\theta_{t,m\in\M})\mathbf{x}_{m}
\Big\vert^2
\bigg\}
\nonumber\\
&
\hspace{-7em}
=
\underbrace{\Ex\bigg\{
\Big\vert
\sum\nolimits_{m\in\M} \sqrt{a_{m}}\mathbf{a}_N^\dag(\theta_{t,m\in\M}) {\qx}_{c,m}
\Big\vert^2
\bigg\}}_{\triangleq P^{\ave}_{\Com}(\theta_{t,m\in\M})}
\nonumber\\
&
\hspace{-7em}
+
\underbrace{
\Ex\bigg\{
\Big\vert
\sum\nolimits_{m\in\M} \sqrt{(1-a_{m})}\mathbf{a}_N^\dag(\theta_{t,m\in\M}) {\qx}_{r,m}
\Big\vert^2
\bigg\}
}_{\triangleq P^{\ave}_{\sen}(\theta_{t,m\in\M})},
\end{align}
where the expectation is taken over the transmitted signal $\qx_m$ and the small-sale fading. Note that the second term denotes the desired power pattern for sensing, while the first term represents the power pattern for communication, which causes distortion in the sensing process.

Since the orientation of the sensing beams are different
from that of the communication beams to satisfy the sensing
field of view, we will show that $P^{\ave}_{\Com}(\theta_{t,m\in\M})$ is independent of the angles. However, since all S-APs coherently steer the sensing beams towards multiple sensing zones, the sensing pattern for a given zone $l$ experiences sidelobe levels and consists of desired power pattern and distortion power pattern related to other sensing zones. Therefore, for a target angle $\theta_{t,ml}$, we  have
\begin{align}
    P^{\ave}_{\sen}(\theta_{t,ml}) = P^{\ave}_{\sen, \DS}(\theta_{t,ml})
    +P^{\ave}_{\sen, \DST}(\theta_{t,ml}),
\end{align}
where $P^{\ave}_{\sen, \DS}(\theta_{t,ml})$ is the desired sensing pattern and $P^{\ave}_{\sen,\DST}(\theta_{t,ml})$ denotes the distortion sensing pattern.

We now define the MASR metric to evaluate the performance of the sensing task. The MASR metric shares similarities with the sensing beampattern gain metric at a target direction $\theta_{t,ml}$~\cite{Stoica:TSP:2007,Hua:TWC:2024,Ren:TCOM:2023} or sensing SINR at the $l$-th target~\cite{Zhenyao:JSAC:2023}. The MASR characterizes the illumination power achieved at the location of the target of interest. Compared to the conventional ISAC systems in~\cite{Stoica:TSP:2007,Hua:TWC:2024,Ren:TCOM:2023}, in our proposed design, communication signals are not incorporated into the sensing function and may cause distortions in the sensing beampattern gain. To address this, we isolate the distortion component and define MASR as the ratio of the main sensing beampattern lobe to the distortion lobes. 

It is desirable to have the  average sensing power pattern distortion as well as communication distortion, i.e.,  $P^{\ave}_{\Com}(\theta_{t,m\in\M})$ as small as possible to confine the pattern distortion. For illuminating a target angle $\theta_{t,ml}$, it is desirable that $P^{\ave}_{\sen, \DS}(\theta_{t,ml})$ is higher than the sum of $P^{\ave}_{\sen, \DST}(\theta_{t,ml})$ and $P^{\ave}_{\Com}(\theta_{t,m\in\M})$ by a certain minimum sensing level $\kappa$, which is referred to as the $\mathrm{MASR}_l$~\cite{Liao:TWC:2024}, i.e.,
\vspace{0.3em}
\begin{align}~\label{eq:MASRl}
     \mathrm{MASR}_l \triangleq \frac{P^{\ave}_{\sen, \DS}(\theta_{t,ml})}{P^{\ave}_{\Com}(\theta_{t,m})+ P^{\ave}_{\sen \DST}(\theta_{t,ml})}.
\end{align}

Now, we proceed to derive explicit expressions of different power pattern terms in~\eqref{eq:MASRl}. By using~\eqref{eq:xcm} and considering that $\tmkZFcom$ and $\tmkMRcom$ are independent, we have
\vspace{0.3em}
\begin{align}
    &P^{\ave}_{\Com}(\theta_{t,m})&\nonumber\\
&=\!\sum\nolimits_{k\in\KK} \!\Big(\!\sum\nolimits_{m\in\ZK}
\!\!
    {a_m\rho\etamkc}  \Ex\Big\{\Big\vert\qa_N^\dag(\theta_{t,m})\tmkZFcom\Big\vert^2\Big\} 
    \nonumber\\
    &\hspace{2em}+\sum\nolimits_{m\in\MK}
    \!\!\!
   {a_m\rho\etamkc}  
   \Ex\bigg\{\Big\vert\qa_N^\dag(\theta_{t,m})\tmkMRcom\Big\vert^2 \!\Big\}\!\Big).
\end{align}
By invoking Lemma~\ref{lemma:wPZF} in Appendix~\ref{Apx:wPZF}, we obtain
\begin{align}~\label{eq:Pavcom}
P^{\ave}_{\Com}(\theta_{t,m})
&=\rho\sum\nolimits_{m\in\MM} 
       \sum\nolimits_{k\in\KK}
\!\!a_{m}
\etamkc \gammk\nu_{mk},
\end{align}
where $\nu_{mk}\triangleq\frac{\dkks}{N-\vert \Sm \vert}\!+N\dkkw$, while
   \begin{align}
    \delta_k^{\Ks}=
    \begin{cases}
     1 & \text{if } k \in\Ks,\\
      0 & \text{otherwise}.
    \end{cases}
    \quad \delta_k^{\Kw}=
    \begin{cases}
     1 & \text{if } k \in \Kw,\\
      0 & \text{otherwise},
    \end{cases}
\end{align} 
denote the binary variables to indicate the group assignment for each AP $m$ and downlink UE $k$.

Moreover, $P^{\ave}_{\sen}(\theta_{t,ml})$ in~\eqref{eq:pavet1} can be obtained as
\begin{align}~\label{eq:T2}
    &P^{\ave}_{\sen}(\theta_{t,ml})=\nonumber\\
    &\hspace{0.0em}
\rho\!\sum\nolimits_{m\in\M}\!\sum\nolimits _{l'\in\LL}\! (1\!-\!a_{m})\etamlps \big\vert \mathbf{a}_N^\dag(\theta_{t,ml})\qa_N(\theta_{t,ml'}) \big|^2\!.
\end{align}

We notice that $P^{\ave}_{\Com}(\theta_{t,m})$ is independent of angles.  Now, we consider the sensing performance associated with target $l$, i.e., at angle $\theta_{t,ml}$ for all $m\in\M$. By using~\eqref{eq:T2}, $ P^{\ave}_{\sen}(\theta_{t,ml})$ can be expressed as
\begin{align}~\label{eq:Pavsen}
    &P^{\ave}_{\sen}(\theta_{t,ml})=
\underbrace{\rho N^2\sum\nolimits_{m\in\M} (1-a_{m})\eta^{s}_{ml}}_{P^{\ave}_{\sen, \DS}(\theta_{t,ml})} 
\\
&\hspace{-1em}+\underbrace{\rho\sum\nolimits_{m\in\M}\sum\nolimits _{l'\in\LL\setminus l}
(1-a_{m})\etamlps \big\vert \mathbf{a}_N^\dag(\theta_{t,ml})\qa_N(\theta_{t,ml'}) \big|^2}_{P^{\ave}_{\sen, \DST}(\theta_{t,ml})},\nonumber
\end{align}
where the first term denotes the desired power pattern for sensing at the angles related to target $l$ and the second term is the average sensing power pattern distortion related to the other sensing zones.

Accordingly, by substituting~\eqref{eq:Pavcom} and~\eqref{eq:Pavsen} into~\eqref{eq:MASRl}, and then incorporating the path loss coefficient between the APs and target zone $l$,  the $\mathrm{MASR}_l$ for the $l$th  sensing zone is defined by~\eqref{eq:MASR1} at the top of the next page. Note that, in general, in the MASR, we need to choose angles that maximize the average sensing power pattern distortion. However, this leads to an intractable MASR form. To alleviate such difficulty and make the MASR more amenable to further design, we choose $\theta_{t,ml}$ for computing the sensing power pattern distortion in MASR. This is reasonable since we assume that we know the sensing zone $l$, and thus, we only need to search some directions around the angles of target $l$.

\begin{figure*}
   \begin{align}~\label{eq:MASR1}
   \mathrm{MASR}_l(\mathbf{a},  \boldsymbol{\eta}^{c}, \boldsymbol{\eta}^{s}) &\!=\! 
  \frac{\sum\nolimits_{m\in\M}
(1-a_{m})N^2\eta^{s}_{ml}}
   {\sum\nolimits_{m\in\M} 
   \sum\nolimits_{k\in\KK}a_m \etamkc\gammk
   \nu_{mk}
\!+ \!
\sum\nolimits_{m\in\M}\sum\nolimits _{l'\in\LL\setminus l} (1\!-\!a_{m})\etamlps \big\vert \qa_N^\dag(\theta_{t,ml})\qa_N^\dag(\theta_{t,ml'}) \big|^2 }.
\end{align}
\hrulefill
\vspace{-5mm}
\end{figure*}

\vspace{-0.5em}
\section{Asymptotic Analysis }~\label{sec:Asy}
The system comprises a total of $M$ APs. The APs are divided into two categories: $M_c$ APs designated for communication activities and $M_s$ APs dedicated to sensing tasks, ensuring the total AP count remains as $M=M_c+M_s$.
With the application of $\MR$ precoding across all communication APs,  we can pursue a detailed assessment of $\MR$ precoding's impact on the effectiveness of CF-mMIMO ISAC systems. In addition, we assume that equal power control is applied at the  C-APs and S-APs, where
 \vspace{-0.1em}
\begin{align}\label{eq:eta}
\bar{\eta}^{c}_{mk}=\frac{1}{N \sum_{k'\in\KK}\gamma_{mk'}}, ~ \forall k,\quad
\bar{\eta}^{s}_{ml}=\frac{1}{N L}  ~ \forall l.
\end{align}

We analyze the performance of two case studies: \textit{i)}  $M$ is large and $N$ is fixed, and; \textit{ii)}  $N$ is large and $M$ is fixed. 

\vspace{-0.5em}
\subsection{Case I ($M_s, M_c\to\infty$  and $N$ is fixed)}
In this section, we explore a scenario where both the numbers of APs designated for $M_s$  and $M_c$  approach infinity. Importantly, during this expansion, we maintain a constant ratio between $M_c$ and $M_s$, symbolized as $c = \frac{M_c}{M_s}$. This examination sheds light on the system's asymptotic performance characteristics, when both the sensing and communication resources are scaled infinitely but proportionally.
A critical aspect of this scenario involves the transmit power scaling of each AP. As $M_c$ becomes very large, the transmit power per AP is adjusted according to $\rho=\frac{E}{M_c^2}$, where $E$ is fixed.  For the sake of notation simplicity, define $\bar{E}\triangleq \frac{a_m E}{ \sumgammk}$. By applying~\eqref{eq:eta}~into~\eqref{eq:yk:utf}, we have
 \begin{align}\label{eq:Myk}
&y_k= 
\sum\nolimits_{m\in{\MM}}  \sqrt{ \frac{\bar{E}}{ N }} \frac{1}{M_c}\gtmk \qt^{\MRcom}_{mk}  x_{c,k}
\\
&\hspace{1em}
+\!\!\sum\nolimits_{k'\in\KK \setminus k}
\!\sum\nolimits_{m\in \MM}
\!\!\sqrt{\! \frac{\bar{E}}{N }}\frac{1}{M_c} \gtmk \qt^{\MRcom}_{mk'}    \!x_{c,k'}
\nonumber\\
&\hspace{1em}
+\!\!\sum\nolimits_{m\in\mathcal{M}} \sum\nolimits_{l\in\LL} \sqrt{\frac{\left(1-a_m\right)E}{NL} }\frac{1}{cM_s}\gtmk\tmlsen x_{r}+{n}_k.
\nonumber
\end{align}
 
Then, by applying Tchebyshev’s theorem we have\footnote{Let $X_1,\ldots,X_n$ be independent RVs, such that~$\Ex\Big\{X_i\}=\bar{x}_i $ and $\text{Var} \leq c \le \infty. $ Then, Tchebyshev's theorem states $ \frac{1}{n} \sum_{n'=1}^{n} X_{n'} - \frac{1}{n} \sum_{n'} \bar{x}_{n'} \xrightarrow[{n\to \infty }]{P} 0$~\cite{cramer2004random}. }  
  \begin{align}~\label{eq:asym1}
&\sum\nolimits_{m\in{\MM}}  \sqrt{ \frac{\bar{E}}{N }} \frac{1}{M_c}\gtmk \qt^{\MRcom}_{mk}\nonumber\\
&-  \sum\nolimits_{m\in\M}  \sqrt{ N \bar{E}} \frac{\gammk}{M_c} \xrightarrow[{M_c \to \infty }]{P}0.
\end{align}  

Now, we turn our attention on the interference terms. For inter-user interference, by noting that the zero-mean channel vector $\qg_{m k}$ is independent of $\qt^{\MRcom}_{mk'}$ for $k'\neq k$,  we have 
\vspace{-0.0em}
  \begin{align}~\label{eq:asym2}
\sum\nolimits_{m\in \MM} \sqrt{ \frac{ \bar{E}}{ N }} \frac{1}{M_c}\gtmk \qt^{\MRcom}_{mk'}\xrightarrow[{M_c\to \infty }]{P}0.
\end{align}  

Moreover, for the deterministic $\tmlsen$, we have 
\begin{align}~\label{eq:asym3}
   \sum\nolimits_{m\in\MM} \sum\nolimits_{l\in\LL} \sqrt{\frac{\left(1-a_m\right)E}{ NL} }\frac{1}{cM_s}\gtmk\tmlsen \xrightarrow[{M_s \to \infty }]{P}\!0.
\end{align}

Accordingly, by using~\eqref{eq:asym1},~\eqref{eq:asym2}, and~\eqref{eq:asym3},  we have
\begin{align}\label{eq:M}
{y_k} -\bigg( \!\! \sum\nolimits_{m\in\M} \! \sqrt{ N \bar{E}} \frac{\gammk}{M_c} \bigg)x_{c,k}-n_k
\xrightarrow[{M_c,M_s \to \infty }]{P}\!0.
\end{align}
The result in~\eqref{eq:M} indicates that when both $M_c$  and $M_s$   approach infinity, the received signal contains the desired signal plus noise. This indicates that  inter-user interference and interference caused by S-APs fade, enhancing the signal quality. Most importantly, as $M$ grows large, while the ratio of $M_c$  and $M_s$ is kept fixed, the transmit power at each AP  can be made inversely proportional to the square of the number of APs with no degradation in performance.

In practical circumstances, although $M_c$ will not be infinite, our asymptotic analysis offers valuable insights into the system's scalability and energy efficiency as $M_c$ increases. It suggests that by deploying a large number of APs, the per-AP power can be significantly reduced while still yielding the desired communication performance, potentially leading to substantial power savings.

\vspace{-1em}
\subsection{Case II ($N\to\infty$ and $M_c$ and $M_s$ are fixed)}
We explore the scenario where $N$ approaches infinity, while $M_c$ and $M_s$ remain constant. The scaling of the transmit power for each AP is adjusted according to $\rho = \frac{E}{N}$, where $E$ remains a fixed value. 

By implementing the power scaling rule $\rho = \frac{E}{N}$,  considering ~\eqref{eq:yk:utf} and substituting $\qg_{mk}=(\hat{\qg}_{m k} +\tilde{\qg}_{m k})$, the desired received signal at the UE can be characterized as follows
 \begin{align}\label{eq:lyk}
 &\sum\nolimits_{m\in{\MM}}\!\! \frac{ \sqrt{ \bar{E}}}{N}(\hat{\qg}_{m k} \!+\!\tilde{\qg}_{m k})^{\dag}\mathbf{\hat{g}}_{mk} 
 \!\buildrel a.s. \over\longrightarrow \!\sum\nolimits_{m\in{\MM}} \sqrt{\bar{E}}\gammk,
\end{align}
where we have used the results on very long random vectors~\cite{Ngo:TCOM:2013}.

For the interference terms, noticing that $\gtmk$ and $\qt^{\MRcom}_{mk'}$ are independent zero-mean vectors,  we have
\begin{subequations}~\label{eq:IUI}
\begin{align}
 &\sum\nolimits_{m\in \MM}  \sqrt{ \bar{E}} \frac{1}{N}\gtmk \qt^{\MRcom}_{mk'} \buildrel a.s. \over\longrightarrow 0,
 \\
 &\sum\nolimits_{m\in\mathcal{M}} \sum\nolimits_{l\in\LL} \sqrt{\frac{\left(1-a_m\right)E}{L} }\frac{1}{N}\gtmk\tmlsen \buildrel a.s. \over\longrightarrow 0.
\end{align}
\end{subequations}

Then, by using~\eqref{eq:lyk} and~\eqref{eq:IUI} we have
\begin{align}
y_k&-\sum\nolimits_{m\in{\MM}}\sqrt{\bar{E} }\gammk x_{c,k} 
- n_k \buildrel a.s. \over\longrightarrow0,~\text{as}~N \rightarrow \infty,
\end{align}
where $\buildrel a.s. \over \rightarrow$ denotes almost sure convergence. Akin to colocated mMIMO systems, a substantial power gain (i.e. $\rho = \frac{E}{N}$) can be achieved, and the channels between the UEs and APs tend toward orthogonality as $N$ grows large ($ N\to\infty$), elucidating the profound impact of deploying an extensive number of antennas in the CF-mMIMO ISAC space. This tendency towards channel orthogonality inherently mitigates the interference and enhances the signal quality, thereby reinforcing the premise that increasing the antenna count per AP not only avails of the advantages seen in traditional mMIMO architectures, but also significantly boosts the overall system performance and SE within the CF-mMIMO ISAC paradigm.

\section{Proposed Design Problems and Solution}~\label{sec:power}
In this extended section, we formulate and solve the problem of joint AP mode selection and power allocation to provide fairness across the UEs, subject to quality of service (QoS) requirements of the sensing zone. More specifically, we aim to optimize the AP operation mode selection vector ($\qa$) and power control coefficients ($\ETAC_{mk}, \ETAS_{ml}$) to maximize the minimum per-UE SE subject to a prescribed MASR level for the target detection and transmit power constraints at the APs.  

In order to further simplify the optimization problem, we propose a greedy-based algorithm for AP mode selection and power control design. More specifically,  the joint optimization problem is decomposed into two sub-problems: 1) AP mode selection, which is performed via a greedy algorithm, and 2) power control design for fixed AP modes.

Note that AP $m$ is required to meet the average normalized
power constraint, i.e., $\Ex\left\{\|\qx_{c,m}\|^2\right\}\leq \rho$. By invoking~\eqref{eq:powerconst}, and noticing that $\Ex\Big\{\big\Vert \qt^{\ZFcom}_{mk}\big\Vert^2\Big\} =\frac{\gammk}{N-\vert\Sm\vert}$, we have the following per-AP power constraint
\vspace{0.4em}
\begin{align}
\label{downlink:power:cons}
&\sum\nolimits_{k\in\Sm} \eta^{c}_{mk}\frac{\gammk}{N-\vert\Sm\vert} +
\sum\nolimits_{k\in\Wm} N\gammk \eta^{c}_{mk}  \leq a_m.
\end{align}

\subsection{Joint AP Mode Selection and Power Allocation Design}
In this subsection, we formulate and solve the problem of joint AP mode selection and power control design. Define $\mathbf{a} \triangleq \{a_1,\ldots,a_M\}$, $\ETAC \triangleq\{\eta^c_{m1},\ldots,\eta^{c}_{mK}\}$ and $\ETAS\triangleq\{\eta_{m1}^s,\ldots,\eta^{s}_{mL}\}$ for all $m\in\M$. The optimization problem is then formulated as\footnote{ In  sensing-priority applications, maximizing the beampattern gain and minimizing the mean-squared cross-correlation pattern are interesting directions for future research~\cite{Liu:TSP:2020}.}
\vspace{0.4em}
\begin{subequations}\label{P:Extended_PZF}
    \begin{align}
        \text{\textbf{(P1):}}~\underset{\mathbf{a}, \ETAC, \ETAS}{\max}  \min_{k\in\mathcal{K}} \hspace{1em}\, &\SINR_k (\mathbf{a}, \ETAC, \ETAS) \label{eq:max_min_snr}\\
        \mathrm{s.t.} \hspace{1em}\, & \mathrm{MASR}_l(\mathbf{a}, \ETAC, \ETAS) \geq\kappa,~\forall l\in\LL, \label{eq:Extended_PZF:ct1}\\
         &\sum\nolimits_{k\in\KK} \eta^{c}_{mk}\gammk\nu_{mk} \leq a_m, \forall m\in\mathcal{M}, \label{eq:etamk:ct31}\\
        &\sum\nolimits_{l\in\LL}\!\eta^{s}_{ml}\! \leq\! \frac{ (1 \!- \!a_m)}{N}, \forall m\in\mathcal{M}, \label{eq:etam:ct31}\\
        & a_m \in \{ 0,1\}, \label{eq:am:binary_strategies1}
    \end{align}
\end{subequations}
where the first constraint~\eqref{eq:Extended_PZF:ct1} is to explicitly control the level of the MASR for all sensing zones, while the second and third constraint control the transmit power at the C-APs and S-APs, respectively. 

Before proceeding, we introduce an auxiliary variable $t =  \min_{k\in\KK} \SINR_k (\mathbf{a}, \ETAC, \ETAS) $. We further relax the binary constraint~\eqref{eq:am:binary_strategies1}. To address the integer binary constraint, it is straightforward to see that $a_m\in\{0,1\}$ is equivalent to $a_m=a_m^2$, $a_m\in[0,1]$. On the other hand, it holds true that $a_m^2\leq a_m$, for $a_m\in[0,1]$. Following~\cite{Che:TWC:2014}, we relax the binary $a_m$ to $a_m\in[0,1]$ and introduce a penalty term with parameter $\lambda$ in the objective function to enforce $a_m=a_m^2$, thus making $a_m$ binary. The penalty parameter $\lambda$ signifies the relative importance of recovering binary values for $\qa$ over fairness maximization~\cite{Che:TWC:2014}.

We notice that since at the optimal point $a_m \in \{ 0,1\}$, we can replace $a_m$ with $a_m^2$ in~\eqref{eq:etamk:ct31} and~\eqref{eq:etam:ct31} to accelerate the convergence speed of the optimization problem. To make the objective function and new constraint convex, we employ SCA and apply the following inequality
\vspace{0.4em}
\begin{equation}~\label{eq:x2}
x^{2}\geq x_{0}(2x-x_{0}).
\end{equation}

Accordingly, the optimization problem \eqref{P:Extended_PZF}  is reformulated as follows:
\vspace{0.4em}
\begin{subequations}\label{P:Extended_t}
    \begin{align}
        \text{\textbf{(P2):}}~\underset{\qa, \ETAC, \ETAS, t}{\max} \hspace{1em}& t\!-\!\lambda \sum\nolimits_{m\in\MM}\!\!a_m-a_m^{(n)}\Big(2a_m\!-\!a_m^{(n)}\Big) \\
        \mathrm{s.t.} \hspace{1em}\,  
        & \SINR_k (\mathbf{a}, \ETAC, \ETAS)\geq t,~  
                 \forall k\in\KK, \label{P:Extended_tb}
        \\ & \mathrm{MASR}_l(\mathbf{a}, \ETAC, \ETAS) \geq\kappa,~\forall l\in\LL, \label{P:Extended_tc}\\
        &\hspace{0em}\sum\nolimits_{k\in\KK} \eta^{c}_{mk}\gammk\nu_{mk}
        \nonumber\\
        &\hspace{2em}
        \leq {a_m^{(n)}\Big(2a_m\!-\!a_m^{(n)}\Big)}, \forall m\in\mathcal{M}, \label{eq:etamk:ct3}\\
        & \sum_{l\in\LL}\eta^{s}_{ml} \leq \frac{ (1 - a_m^2)}{N},~\forall m\in\mathcal{M}, \label{eq:etam:ct3}\\
        & 0\leq a_m \leq 1,~\forall m\in\mathcal{M}. \label{eq:am:binary_strategies}
    \end{align}
\end{subequations}

The non-convex nature of constraints \eqref{P:Extended_tb} and \eqref{P:Extended_tc} make the resulting problem  non-convex. To address the non-convexity, we employ the method of SCA. To deal with the non-convexity of~\eqref{P:Extended_tb}, we first define 
\vspace{0.4em}
\begin{subequations}
  \begin{align}
 \varrho_{mk}&\triangleq \zmzk\frac{(\betmk-\gammk)}{N-\vert\Sm\vert}+\zmmk N\betmk,\\
 f_{mk}&\triangleq \zmzk + N\zmmk.
\end{align}  
\end{subequations}
Now, constraint~\eqref{P:Extended_tb} can be expressed as
\vspace{0.4em}
\begin{align}~\label{eq:dLSNIRconstraint}
&\frac{
\Big(\sum_{m\in\MM}  \sqrt{ a_m  \rho\etamkc} \gammk f_{mk} \Big)^2
                 }
                 {t} \geq
                 \nonumber\\
                 &                  
   \rho \sum\nolimits_{m\in\MM} a_m\Big(\sum\nolimits_{k'\in\KK }  \etamkpc\gamma_{mk'}  
    \varrho_{mk}  
    - \sum\nolimits _{l\in\LL} \etamls N  \betmk \Big)\nonumber\\
&
    +  \rho\sum\nolimits_{m\in\M}\sum\nolimits _{l\in\LL} \etamls N  \betmk +1.
    \end{align}
For ease of description, let us denote   
\vspace{0.2em}
\begin{align}
    \mu_{mk}&\triangleq\sum\nolimits_{k'\in\KK }  \etamkpc\gamma_{mk'}  \varrho_{mk} - \sum\nolimits _{l\in\LL} \etamls N  \betmk.
\end{align}
To this end, ~\eqref{eq:dLSNIRconstraint} is equivalent to
\vspace{0.2em}
\begin{align}~\label{eq:dLSNIRconstraint2}
&\frac{
\Big(2\sum_{m\in\MM}  \sqrt{ a_m  \etamkc} \gammk f_{mk} \Big)^2
                 }
                 {t} \!+\!\sum\nolimits_{m\in\MM} (a_m\!-\!\mu_{mk})^2\nonumber\\
                 &\hspace{0em} 
                 \!\geq\!
                 \sum\nolimits_{m\in\MM} (a_m\!+\!\mu_{mk})^2\!\!
     +\!  4\sum\nolimits_{m\in\M}\sum\nolimits _{l\in\LL} \etamls N  \betmk \!+\!4/\rho.
    \end{align}
It is clear that we need to find a concave lower bound of
the left-hand side of the above inequality.  To this end, we note that the function $x^{2}/y$ is convex for $y>0$, and thus, the following inequality holds
\vspace{-0.3em}
\begin{align}~\label{eq:x2/y}
\frac{x^{2}}{y}\geq
\frac{x_{0}}{y_{0}}\bigl(2x-\frac{x_{0}}{y_{0}}y\bigr),
\end{align}
which is obtained by linearizing $x^{2}/y$ around $x_{0}$ and $y_{0}$. Let us now define 
\begin{align}
   q_k^{(n)}\triangleq \frac{2\sum_{m\in\MM}  \sqrt{ a_m^{(n)}  (\etamkc)^{(n)}} \gammk f_{mk}} {t^{(n)}}. 
\end{align}
From the inequality~\eqref{eq:x2/y}, we can recast~\eqref{eq:dLSNIRconstraint2} as
\begin{align}~\label{eq:dLSNIRcf}
&     q_k^{(n)}\Big(4\sum\nolimits_{m\in\MM}  \sqrt{ a_m  \etamkc} \gammk f_{mk}\!-\! q_k^{(n)} t \Big)
\nonumber\\
                   &+\!\sum\nolimits_{m\in\MM}\! \big(a_m^{(n)}\!-\!\mu_{mk}^{(n)}\big)\Big(2\big(a_m-\mu_{mk}\big)-\big(a_m^{(n)}-\mu_{mk}^{(n)}\big)\Big)
                   \nonumber\\
                   &\geq
                  \sum\nolimits_{m\in\MM} (a_m\!+\!\mu_{mk})^2
        +\!
        4\!\sum\nolimits_{m\in\M}\!\sum\nolimits _{l\in\LL}\! \etamls N  \betmk \!+\! 4/\rho,
    \end{align}
where we have used~\eqref{eq:x2} and replaced $x$ and $x_{0}$ by $a_m-\mu_{mk}$ and $a_m^{(n)}-\mu_{mk}^{(n)}$,
respectively.

Now, we focus on~\eqref{P:Extended_tc}, which can be expressed as
\vspace{0.1em}
\begin{align}~\label{eq:MASR2}
& {\sum\nolimits_{m\in\M}
N^2\eta^{s}_{ml}}
\geq
\nonumber\\
&
\sum\nolimits_{m\in\M}
a_{m} \Big(N^2\eta^{s}_{ml}
 +  \kappa 
   \sum\nolimits_{k\in\KK}\etamkc\gammk
   \nu_{mk}
   \nonumber\\
   &
\!- \!
\kappa\sum\nolimits _{l'\in\LL\setminus l}\etamlps \big\vert \mathbf{a}_N^\dag(\theta_{t,ml})\qa_N(\theta_{t,ml'}) \big|^2\Big)
   \nonumber\\
   &
\!+ \!
\kappa\sum\nolimits_{m\in\M}\sum\nolimits _{l'\in\LL\setminus l} \etamlps \big\vert \qa_N^\dag(\theta_{t,ml})\qa_N(\theta_{t,ml'}) \big|^2.
\end{align}
For ease of description let us denote
\begin{align}
    \omega_{ml}&\triangleq 
    \Big(N^2\eta^{s}_{ml}
 +  \kappa 
   \sum\nolimits_{k\in\KK}\etamkc\gammk
   \nu_{mk} \nonumber\\
   &
\!- \!
\kappa\sum\nolimits _{l'\in\LL\setminus l}\etamlps \big\vert \mathbf{a}_N^\dag(\theta_{t,ml})\qa_N(\theta_{t,ml'}) \big|^2\Big).
\end{align}
Accordingly,~\eqref{eq:MASR2} can be written as
\begin{align}~\label{eq:MASR3}
& 4\sum\nolimits_{m\in\M}
N^2\eta^{s}_{ml}
+
\sum\nolimits_{m\in\M}
(a_{m} -\omega_{ml})^2
\geq
\nonumber\\
&
4\kappa\sum\nolimits_{m\in\M}\sum\nolimits _{l'\in\LL\setminus l} \etamlps \big\vert \mathbf{a}_N^\dag(\theta_{t,ml})\qa_N(\theta_{t,ml'}) \big|^2\nonumber\\
&+\sum\nolimits_{m\in\M}
(a_{m} +\omega_{ml})^2.
\end{align}
To this end, by using~\eqref{eq:x2}, we obtain the concave lower bound of the left-hand side of the above inequality. Then, we get
\vspace{0.1em}
\begin{align}~\label{eq:MASR4}
& 4\sum\nolimits_{m\in\M}
N^2\eta^{s}_{ml}
+
\sum\nolimits_{m\in\M}
\big(a_{m}^{(n)} -\omega_{ml}^{(n)}\big) \nonumber\\
&\times\!\Big( 2(a_{m}\!-\!\omega_{ml}) \!-\! (a_{m}^{(n)} -\omega_{ml}^{(n)})\Big)
\!\geq\!
\sum\nolimits_{m\in\M}
(a_{m} +\omega_{ml})^2
   \nonumber\\
   &\hspace{0em}
\!+ \!
4\kappa\sum\nolimits_{m\in\M}\sum\nolimits _{l'\in\LL\setminus l} \etamlps \big\vert \mathbf{a}_N^\dag(\theta_{t,ml})\qa_N(\theta_{t,ml'}) \big|^2.
\end{align} 

Now, the convex optimization problem is given as~\eqref{JAPpower_final}, at the top of the next page. Problem (\textbf{P3})  is a convex optimization problem and can be efficiently solved using CVX~\cite{cvx}. In \textbf{Algorithm~\ref{alg1}}, we outline main steps to solve problem (\textbf{P3}), where $\widetilde{\qx} \triangleq \{\qa, \ETAC, \ETAS\}$ and $\widehat{\mathcal{F}} \triangleq\{\eqref{P4:a},~\eqref{P4:b},~\eqref{P4:c},~\eqref{P4:d},~\eqref{P4:e}\}$  is a convex feasible set. Starting from a random point $\widetilde{\qx}\in\widehat{\mathcal{F}}$, we solve \eqref{JAPpower_final} to obtain its optimized solution $\widetilde{\qx}^*$, and use $\widetilde{\qx}^*$ as an initial point in the next iteration. The proof of this convergence property uses similar steps as  the proof of \cite[Proposition 2]{vu18TCOM}, and hence, is omitted herein due to lack of space.
\begin{figure*}[btp]
  \begin{subequations}\label{JAPpower_final}
    \begin{align}
        \text{\textbf{(P3):}}~\underset{\qa, \ETAC, \ETAS, t}{\max} \hspace{1em}& t-\lambda \sum\nolimits_{m\in\MM}a_m-a_m^{(n)}\Big(2a_m\!-\!a_m^{(n)}\Big) \\
        \mathrm{s.t.} \hspace{1em}\, 
&     q_k^{(n)}\Big(4\sum\nolimits_{m\in\MM}  \sqrt{ a_m  \etamkc} \gammk f_{mk}\!-\! q_k^{(n)} t \Big)
                   \!+\!\sum\nolimits_{m\in\MM}\! \big(a_m^{(n)}\!-\!\mu_{mk}^{(n)}\big)\Big(2\big(a_m-\mu_{mk}\big)-\big(a_m^{(n)}-\mu_{mk}^{(n)}\big)\Big)\nonumber\\
                   & \hspace{4em}
                   \geq
                  \sum\nolimits_{m\in\MM} (a_m+\mu_{mk})^2
        +4\sum\nolimits_{m\in\M}\sum\nolimits _{l\in\LL} \etamls N  \betmk + 4/\rho,~\label{P4:a}\\
& 4\sum\nolimits_{m\in\M}
N^2\eta^{s}_{ml}
+
\sum\nolimits_{m\in\M}
\big(a_{m}^{(n)} -\omega_{ml}^{(n)}\big) \Big( 2(a_{m}-\omega_{ml}) - (a_{m}^{(n)} -\omega_{ml}^{(n)})\Big)
   \nonumber\\
   &\hspace{4em}
\geq
\sum\nolimits_{m\in\M}
(a_{m} +\omega_{ml})^2
\!+ \!
4\kappa\sum\nolimits_{m\in\M}\sum\nolimits _{l'\in\LL\setminus l} \etamlps \big\vert \mathbf{a}_N^\dag(\theta_{t,ml})\qa_N(\theta_{t,ml'}) \big|^2,~\label{P4:b}\\
&\sum\nolimits_{k\in\KK} \eta^{c}_{mk}\gammk\nu_{mk} \leq {a_m^{(n)}\Big(2a_m\!-\!a_m^{(n)}\Big)}, \forall m\in\mathcal{M}, ~\label{P4:c}\\
        & \sum_{l\in\LL}\eta^{s}_{ml} \leq \frac{ (1 - a_m^2)}{N}, \forall m\in\mathcal{M}, ~\label{P4:d}\\
        & 0\leq a_m \leq 1. ~\label{P4:e}
    \end{align}
\end{subequations}
\hrulefill
\vspace{-1.5em}
\end{figure*}

\textit{Complexity of Algorithm~\ref{alg1}:}
\textbf{Algorithm~\ref{alg1}} requires solving a series of convex problems \eqref{JAPpower_final}. For ease of presentation, if we let $K = L$, problem \eqref{JAPpower_final}  can be transformed to an equivalent problem that involves $A_v\triangleq (M + 2MK +1)$ real-valued scalar variables, $A_l\triangleq 3M$ linear constraints, $A_q\triangleq 2K $ quadratic constraints. Therefore, the algorithm for solving problem \eqref{JAPpower_final} requires a complexity of $\mathcal{O}(\sqrt{A_l+A_q}(A_v+A_l+A_q)A_v^2)$ in each iteration~\cite{tam16TWC,Mohammadi:JSAC:2023}. In Section~\ref{sec:num}, we will show that this algorithm converges to the optimized solution after a few iterations.

\subsection{Low-Complexity Design}
In this subsection, to reduce the complexity of the joint optimization problem while maintaining acceptable system performance, we propose a low-complexity design. To achieve this, we decompose the main optimization problem into two disjoint sub-problems: 1) AP mode assignment and 2) AP power allocation. In the first phase, AP mode selection is performed using a greedy algorithm under the assumption of equal power allocation among the APs. In the subsequent phase, the power allocation problem is solved for the given AP modes to optimize the power control coefficients at the C-APs and S-APs.

\begin{algorithm}[!t]
\caption{Proposed algorithm for joint AP mode selection and power allocation design (\textbf{JAP-OPA})}
\begin{algorithmic}[1]
\label{alg1}
\STATE \textbf{Initialize}: $n\!=\!0$, 
$\lambda > 1$, a random initial poin $\widetilde{\qx}^{(0)}\!\in\!\widehat{\mathcal{F}}$.
\REPEAT
\STATE Update $n=n+1$
\STATE Solve \eqref{JAPpower_final} to obtain its optimized solution $\widetilde{\qx}^*$
\STATE Update $\widetilde{\qx}^{(n)}=\widetilde{\qx}^*$
\UNTIL{convergence}
\end{algorithmic}
\end{algorithm}
\setlength{\textfloatsep}{0.4cm}

\subsubsection{Greedy AP Mode Selection}  Let $\mathcal{A}_{\sen}$ and $\mathcal{A}_{\Com}$ denote the sets containing the indices of APs operating as radar, i.e., APs with $a_m=0$, and of APs operating in communication mode, i.e., APs with $a_m=1$, respectively.  In addition, $\mathrm{MASR}_l(\mathcal{A}_{\sen}, \mathcal{A}_{\Com})$ and $\SINR_k(\mathcal{A}_{\sen}, \mathcal{A}_{\Com})$ underline the dependence of the sensing MASR (of the $l$-th zone)  and received SINR (of the $k$-th UE) on the different choices of AP mode selection.  Our greedy algorithm of AP mode selection is shown in \textbf{Algorithm~\ref{alg:Grreedy}}. In this algorithm $\Pi_m= \min_{k\in\KK}\SINR_k( \mathcal{A}_{s},\mathcal{A}_{\Com} \bigcup m)$ represents the minimum SINR among the UEs, calculated using~\eqref{eq:dLSNIRf} under equal power control for C-APs and S-APs. The entries of $\qa$ are fixed values of $0$ and $1$, determined according to the sets $\mathcal{A}_{\sen}$ and $\mathcal{A}_{\Com}$.

To guarantee the sensing MASR requirement, all APs are initially assigned for sensing operation, i.e., $\mathcal{A}_{\sen}=\mathcal{M}$ and $\mathcal{A}_{\Com}=\emptyset$ (see Remark~\ref{rema:SEzero}). Then, in each iteration, one AP switches into communication operation mode for maximizing the minimum SE (or equivalently SINR),  while the minimum  MASRs, required for target sensing, in all sensing zones are guaranteed. This process continues until there is no more improvement in the minimum SINR across all UEs.  We note that the greedy algorithm may fail to determine the operation modes of the APs for a given network realization. However, the \textbf{JAP-OPA }design can successfully identify the C-AP and S-AP groups. This is because the \textbf{JAP-OPA} design simultaneously optimizes the AP modes and power allocation coefficients at the APs to meet the network requirements.

\begin{algorithm}[!t]
\caption{Greedy AP mode selection}
\begin{algorithmic}[1]
\label{alg:Grreedy} 
\STATE
\textbf{Initialize}: Set  $\mathcal{A}_{\Com}=\emptyset$ and $\mathcal{A}_{\sen}=\mathcal{M}$. Set iteration index $i=0$.
\STATE Calculate $\Pi^{\star}[i]=  \min_{k\in\KK}\SE_k({A}_{\sen}, \mathcal{A}_{\Com})$
\REPEAT
\FORALL{$m \in \mathcal{A}_{\sen}$}
\STATE Set $\mathcal{A}_{s}=\mathcal{A}_{\sen} \setminus m$.\\
\IF{ $\mathrm{MASR}_l(\mathcal{A}_{s},\mathcal{A}_{\Com} \bigcup m)\!\geq\!\kappa$, $\forall l\in\LL$ }
\STATE  Calculate $\Pi_m= \min_{k\in\KK}\SINR_k( \mathcal{A}_{s},\mathcal{A}_{\Com} \bigcup m)$\\
 \ELSE
\STATE Set  $\Pi_m=0$
\ENDIF
\ENDFOR
\STATE Set $\Pi^\star[i+1]= \underset{m\in\mathcal{A}_{\Com}} \max \,\,\Pi_m$\\
\STATE $e=|\Pi^\star[i+1]- \Pi^\star[i]|$ 
\IF{$e \geq e_{\min}$ }
\STATE {Update $\mathcal{A}_{\Com}\!=\!\{\mathcal{A}_{\Com}\bigcup m^{\star}\}$ and $\mathcal{A}_{\sen}\!=\!\mathcal{A}_{\sen}\!\setminus\! m^{\star}$}
\ENDIF
\STATE {Set $i=i+1$}
\UNTIL{ $e < e_{\min}$ }
\RETURN $\mathcal{A}_{\sen}$ and $\mathcal{A}_{\Com}$, i.e., the indices of APs operating in radar mode and communication mode, respectively.
\end{algorithmic}
\end{algorithm}
\setlength{\textfloatsep}{0.25cm}

\vspace{-0.1em}
\subsubsection{Power Allocation} For a given operation mode selection vector ($\qa$), we optimize the power control coefficients ($\ETAC_{mk}, \ETAS_{ml}$) to achieve maximum fairness among the EUs. Therefore, the optimization problem~\eqref{P:Extended_PZF} is reduced to
\vspace{-0.2em}
\begin{subequations}\label{OP:PA}
    \begin{align}
        \text{\textbf{(P4):}}~\underset{\ETAC, \ETAS}{\max}  \hspace{1em}& t \label{eq:max_min_snr}\\
        \mathrm{s.t.} \hspace{1em}\, & 
        \SINR_k (\ETAC, \ETAS) \geq t,\label{OP:PA:a}\\
        &\mathrm{MASR}_l(\ETAC, \ETAS) \geq\kappa,~\forall l\in\LL, \label{OP:PA:b}\\
        &\sum\nolimits_{k\in\KK} \eta^{c}_{mk}\gammk\nu_{mk}  \leq a_m, \forall m\in\mathcal{M}, \label{OP:PA:c}\\
        &\sum\nolimits_{l\in\LL}\eta^{s}_{ml} \leq \frac{ (1 - a_m)}{N}, \forall m\in\mathcal{M}. \label{OP:PA:d}
    \end{align}
\end{subequations}
We note that constraints~\eqref{OP:PA:a} and~\eqref{OP:PA:b} are non-convex. To deal with this non-convexity, we first introduce the auxiliary variables $\xi_{mk}^2 = \eta^{c}_{mk}$.

At the next step, we replace the convex function in the numerator of~\eqref{OP:PA:a} with its concave lower bound, which results in the following constraint
\vspace{-0.1em}
\begin{align}~\label{eq:zkn}
        \frac{ z_k^{(n)}
        \Big( 2\sum_{m\in\MM}  \sqrt{ a_m } \xi_{mk}\gammk f_{mk}-z_k^{(n)}\Big)}
                 {\Delta(\Bxi, \ETAS)}   
        \geq t,
\end{align}
where $\Delta(\Bxi, \ETAS)\triangleq \rho\sum_{m\in\MM}a_m \varrho_{mk}\sum_{k'\in\KK }\xi_{mk'}^2\gamma_{mk'}   \!+\!  \rho\sum\nolimits_{m\in\M}N (1-a_m) \betmk
\sum\nolimits _{l\in\LL} \etamls  \!+\!  1$ and $z_k^{(n)} \triangleq\sum_{m\in\MM}  \sqrt{ a_m } \xi_{mk}^{(n)}\gammk f_{mk}$. Therefore, the power control design problem for fixed AP mode operation is given by
\begin{subequations}\label{OP:PA:final}
    \begin{align}
        \text{\textbf{(P5):}}~\underset{\Bxi, \ETAS, t}{\max}  \hspace{1em}& t \label{eq:max_min_snr}\\
        \mathrm{s.t.} \hspace{1em}\, & 
    ~\eqref{eq:zkn},~\eqref{OP:PA:b}-\eqref{OP:PA:d}, 
    \end{align}
\end{subequations}

The optimization problem in~\eqref{OP:PA:final} is again convex and can be efficiently solved using CVX~\cite{cvx}.

\begin{Remark}~\label{rema:SEzero}
    It is important to note that for a given network realization, it is likely that greedy and/or random AP selection schemes cannot guarantee the sensing requirement in the network design. Hence, we introduce a new metric known as the \textbf{success sensing rate}, defined as the ratio of the number of network realizations with successful sensing for all sensing zones to the total number of channel realizations. For a fair comparison, in each channel realization, if the MASR requirements are not met or the optimization problem of a scheme is infeasible, we set the SE of that scheme to zero.
\end{Remark}

\vspace{-1.2em}
\section{Numerical Examples}~\label{sec:num}
In this section, we verify the correctness of our analytical results and the performance of the proposed algorithms.
\vspace{-1.1em}
\subsection{Large-scale Fading Model and System Parameters}
We assume that the $M$ APs and $K$ UEs are uniformly distributed in a square of $0.5 \times 0.5$ km${}^2$, whose edges are wrapped around to avoid the boundary effects~\cite{Behdad:GC:2022,Mohammadi:JSAC:2023}. The sensing zones are randomly positioned at a height of $30$ meters above the square area. We assume, as in~\cite{Liu:TSP:2020,Dong:TGCN:2023,Zhang:TWC:2024,Chen:TCOM:2024}, that each AP possesses an estimate of the target directions. The large-scale fading coefficients $\beta_{mk}$ are modeled as~\cite[Eq. (37)]{emil20TWC}. The values of the network parameters are $\tau=200$, and $\tau_t=K+L$. We also set the bandwidth $B=50$ MHz
and the noise figure $F = 9$ dB. Thus, the noise power $\Sn=k_B T_0 B F$, where $k_B=1.381\times 10^{-23}$ Joules/${}^o$K is the Boltzmann constant, while $T_0=290^o$K is the noise temperature. Let $\tilde{\rho}= 1$ W,  and $\tilde{\rho}_t = 0.25$~W be the maximum transmit power of the APs and uplink training pilot sequences, respectively. The normalized maximum transmit powers ${\rho}$ and ${\rho}_t$ are calculated by dividing these powers by the noise power. Moreover, we set $\varrho=85$.

To evaluate the performance of the proposed AP mode selection schemes, we hereafter consider random AP selection with optimized power allocation (\textbf{RAP-OPA}). 
In the numerical results, \textbf{JAP-OPA} denotes our proposed design in \textbf{Algorithm~\ref{alg1}}, i.e., joint AP mode selection and power allocation design. Moreover, the greedy AP selection and optimized power allocation design is indicated by~\textbf{GAP-OPA}.  To make a fair comparison, the number of sensing APs in the \textbf{RAP-OPA} design, $M_s$, is chosen to be close to the average optimized values obtained by the proposed \textbf{JAP-OPA} scheme.

\begin{figure}[t]
\centering
\includegraphics[width=0.44\textwidth]{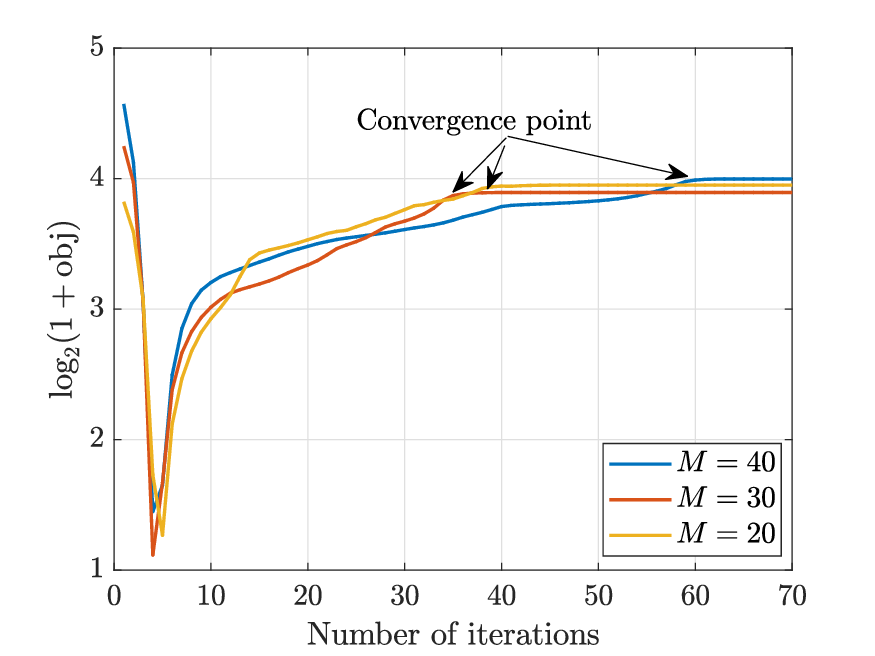}
\vspace{-0.5em}
\caption{\small Convergence behavior of \textbf{Algorithm~\ref{alg1}}  ($MN=480$, $K=4$, $L=2$, $\lambda=10$).}
\label{fig6}
\end{figure}

\begin{table}[t!]
\small
\caption{\small Comparison between the execution time  (in seconds) of \textbf{JAP-OPA} and \textbf{GAP-OPA}}
\vspace{-1em}
\footnotesize
\begin{center}
\renewcommand{\arraystretch}{1}
\begin{tabular}{|p{2cm}|r|r|r|r|}
\hline
\multicolumn{1}{|c|}{\multirow{2}{*}{\textbf{Design}}} & \multicolumn{4}{c|}{\textbf{Number of APs ($M$)}}                                                \\ \cline{2-5} 
\multicolumn{1}{|c|}{}                                   & \multicolumn{1}{c|}{\textbf{$20$}} & \multicolumn{1}{c|}{\textbf{$30$}} & \multicolumn{1}{c|}{\textbf{$40$}} & \multicolumn{1}{c|}{\textbf{$48$}} \\ \hline
\textbf{JAP-OPA}      &\centering173.61     & 269.93     &375.27   & 424.84      \\ \hline
\textbf{GAP-OPA}      &\centering 25.88     &31.15     &35.21   &39.47       \\ \hline

\end{tabular}%
\end{center}
\label{table:execution_time}
\end{table}

\vspace{-0.8em}
\subsection{Results and Discussions}
In Fig.~\ref{fig6}, we compare the convergence rate of \textbf{Algorithm~\ref{alg1}} for different number of APs in the network. To solve~\eqref{JAPpower_final}, we use the convex conic solver MOSEK and set $\lambda=10$. This figure shows the values of $\log_2\big(1+\text{obj}\big)$, where $\text{obj}\triangleq t-\lambda \sum\nolimits_{m\in\MM}a_m-a_m^{(n)}\big(2a_m\!-\!a_m^{(n)}\big)$ for different number of iterations. When the optimization problem begins to converge, i.e., $a_m\approx a_m^{(n)}\in\{0,1\}$, we have $\text{obj}\approx\log_2(1+t)$, where the auxiliary variable is $t =  \min_{k\in\KK} \SINR_k (\mathbf{a}, \ETAC, \ETAS)$. This implies that the value of the objective function converges to the minimum SE per-UE. The transient behavior stems from the penalty term $\lambda \sum\nolimits_{m\in\MM}a_m-a_m^{(n)}\big(2a_m\!-\!a_m^{(n)}\big)$ in the optimization problem. We can see that with a small number of iterations (less than $60$ iterations), \textbf{Algorithm~\ref{alg1}} returns the optimized solution.  Furthermore, it is worth mentioning that the resulting values of the parameters $a_m$ converge to $1$ and $0$ with high accuracy.

\begin{figure}[t]
\centering
\includegraphics[width=0.44\textwidth]{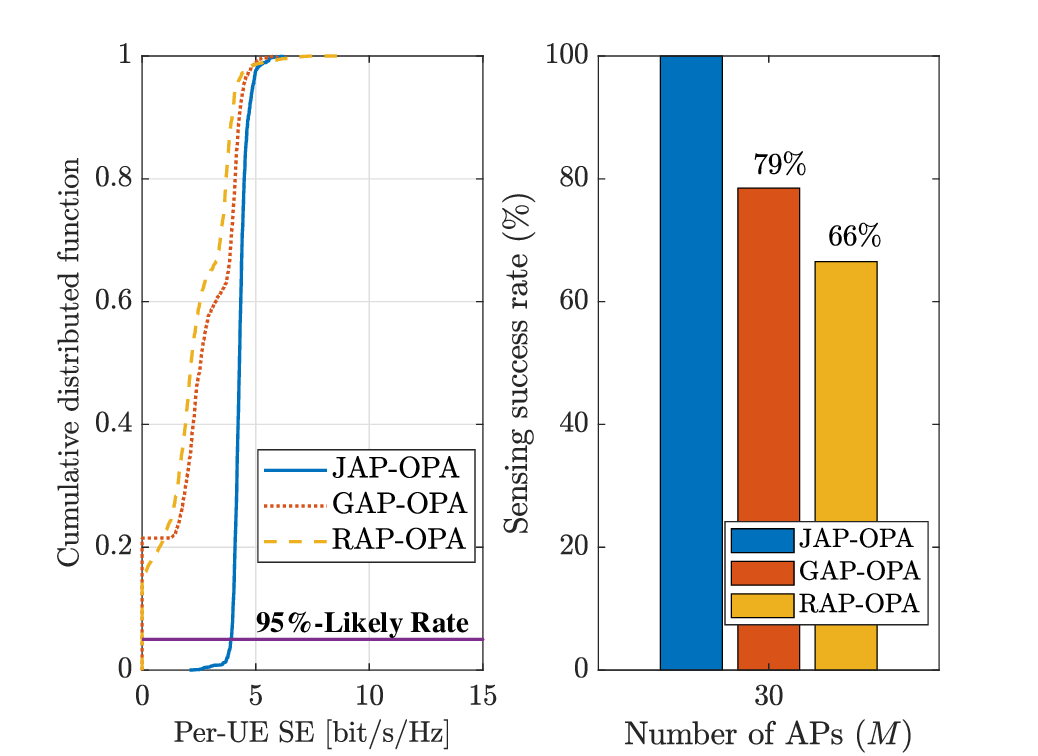}
\vspace{-0.5em}
\caption{\small CDF of the per-UE SE and sensing success rate ($\kappa=8$ dB, $M=30$, $N=16$, $K=4$, and $L=2$).}
\label{fig2}
\end{figure}

Table~\ref{table:execution_time} presents the execution time per iteration for the \textbf{JAP-OPA} and \textbf{GAP-OPA} schemes. The computational resources include a CPU processor, the 12th Gen Intel $\!\!$\textsuperscript{\tiny{\textregistered}} $\!\!$ Core $\!\!\!\!$ \textsuperscript{\tiny{\texttrademark}}  i7-1270P (average clock speed 3.50 GHz) with 16GB RAM.
It can be observed that the execution time for \textbf{RAP-OPA }scales proportionally with an increase in $M$, whereas the execution time for \textbf{GAP-OPA} remains relatively unaffected by changes in $M$. To further reduce the design complexity, a promising future research direction is the development of low-complexity algorithms that leverage artificial intelligence and machine learning to minimize the processing time.

The empirical cumulative distributed functions (CDFs) of the per-UE SE for all scenarios with $\kappa=8$ dB,  $MN=480$, $K=4$ and $L=2$ are plotted in Fig.~\ref{fig2}. The bar chart corresponds to the sensing success rate of different schemes. It can be observed that while \textbf{JAP-OPA} guarantees the sensing requirement for all network realizations, \textbf{RAP-OPA} fails to meet this requirement, with a success rate of $66\%$. In other words, the power allocation problem becomes infeasible in $37\%$ of cases. On the other hand, \textbf{GAP-OPA} yields comparable results, achieving a success rate of approximately $79\%$, which confirms the effectiveness of the proposed greedy algorithm. For the infeasible cases, according to our discussion in Remark~\ref{rema:SEzero}, we set the SE of the UEs to zero. 
By comparing the CDFs of the per-UE SE across all schemes, we observe that \textbf{JAP-OPA} provides a satisfactory fairness among the UEs. Moreover, the SE achieved by \textbf{JAP-OPA} is much higher than that of the \textbf{GAP-OPA} and \textbf{RAP-OPA}. 
\begin{figure}[t]
\centering
\includegraphics[width=0.44\textwidth]{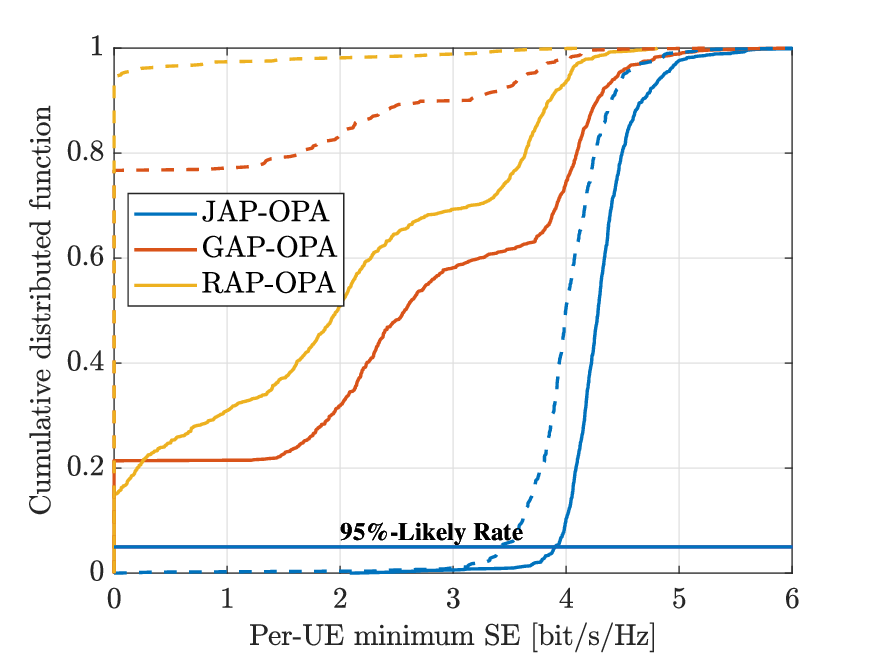}
\vspace{-0.5em}
\caption{\small CDF of the per-UE minimum SE for different schemes and for different number of sensing zones. The dashed lines depict results for $L=3$ while the solid lines show results for $L=2$ ($\kappa=6$ dB, $M=30$, $N=16$, $K=4$).}
\label{fig3}
\end{figure}

Figure~\ref{fig3} shows the worst SE (minimum SE) across the UEs for given realizations of the large-scale fading coefficients and for two different number of sensing zones: $L=3$ (dashed lines) and $L=2$ (solid lines). The proposed \textbf{JAP-OPA} scheme significantly improves the worst SE over the benchmarks and provides $100\%$ sensing success rate for both scenarios. Moreover, we observe that joint optimization can guarantee nearly the same minimum SE per UE across all different channel realizations, while random and greedy AP selection provides varying levels of per-UE SE. For example, when $L=2$, in $60\%$ of channel realizations, the minimum SE obtained via \textbf{GAP-OPA} is less than $3$ bit/s/Hz, while \textbf{JAP-OPA} consistently provides a minimum SE greater than $4$ bit/s/Hz.  Moreover, by increasing the number of sensing zones from $2$ to $3$, the sensing performance of both the \textbf{GAP-OPA} and \textbf{RAP-OPA} severely degraded, resulting in success rates of $27\%$ and $4\%$ for \textbf{GAP-OPA} and \textbf{RAP-OPA}, respectively. These findings confirm the importance of joint AP mode selection and power allocation for both tasks.

Figure~\ref{fig4} shows the average of the minimum SE per-UE  with a fixed number of antenna units in the network versus the number of APs, i.e., we set $MN=480$. By increasing $M$ from $20$ to $48$, the number of antennas per AP, $N$, is reduced from $24$ to $10$. This assumption ensures a fair comparison and allows us to assess the performance of CF-mMIMO under different conditions: 1) a small number of APs with a large number of antennas simulates a more co-located scenario; 2) a small number of antennas with a larger number of APs simulates a more distributed system.  By increasing the number of APs, the proposed \textbf{JAP-OPA} outperforms all other schemes, while \textbf{RAP-OPA} yields the worst performance in terms of SE. Moreover, the sensing success rate of  \textbf{GAP-OPA} and \textbf{RAP-OPA} decreases as the number of APs increases. It is most likely for \textbf{RAP-OPA} to fail when $M>30$.  This behavior is due to the fact that the MASR heavily depends on $N$ (i.e., it is proportional to $N^2$ according to~\eqref{eq:MASR1}). Therefore, the MASR decreases significantly as $M$ increases, and both \textbf{RAP-OPA} and \textbf{GAP-OPA} schemes cannot guarantee to meet the sensing requirement under this condition.  Conversely, the proposed \textbf{JAP-OPA} scheme ensures successful sensing for all network realizations, while delivering much higher per-UE minimum SE compared to the \textbf{RAP-OPA} and \textbf{GAP-OPA} schemes. Interestingly, the SE performance of the  the proposed \textbf{JAP-OPA} scheme is improved by increasing the number of APs. These results reveal the importance of joint AP mode selection and power control design in dense CF-mMIMO networks. 
\begin{figure}[t]
\centering
\includegraphics[width=0.44\textwidth]{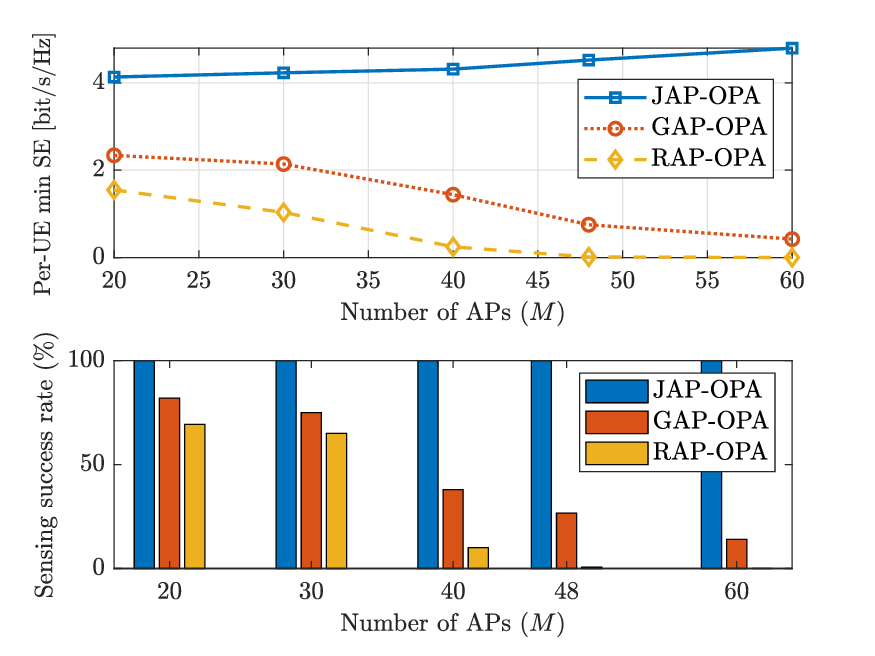}
\vspace{-0.7em}
\caption{\small Average of the per-UE minimum SE and sensing success rate versus the number of APs ($\kappa=8$ dB, $MN=480$, $K=4$, and $L=2$).}
\vspace{-1.4em}
\label{fig4}
\end{figure}

\begin{figure}[t]
\centering
\includegraphics[width=0.44\textwidth]{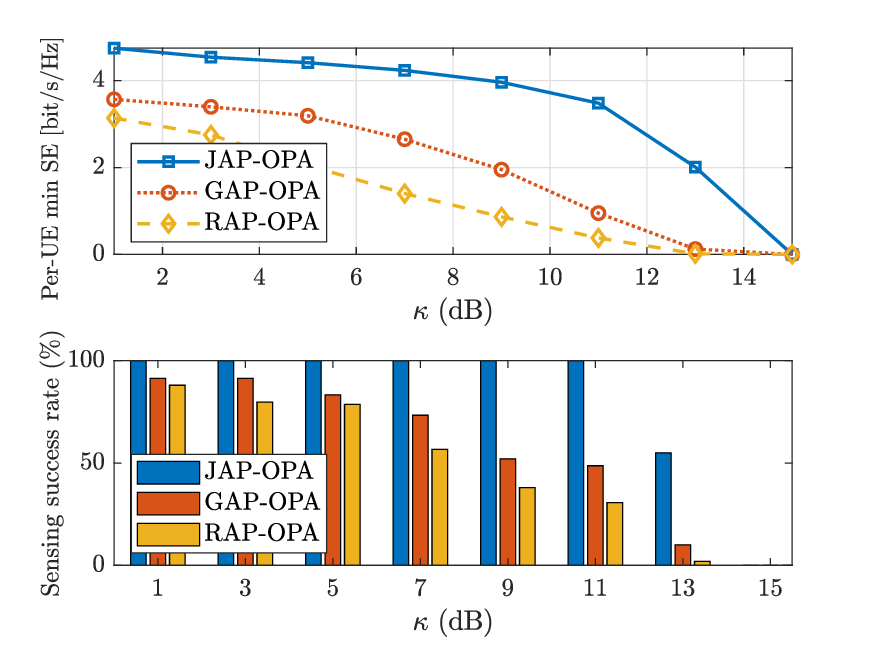}
\vspace{-1em}
\caption{\small Average of the per-UE minimum SE and sensing success rate versus $\kappa$ ($MN=480$,  $M=30$, $K=4$, and $L=2$).}
\vspace{-1em}
\label{fig5}
\end{figure}

Figure~\ref{fig5} illustrates the average of the minimum SE per-UE versus $\kappa$, along with the sensing success rate of different schemes as a function of $\kappa$. As the sensing requirements (i.e., $\kappa$) increase, the sensing success rate of all schemes decreases. Specifically, for moderate values of $\kappa$ (i.e., $6 < \kappa < 10$), the minimum SE per UE for both \textbf{RAP-OPA} and \textbf{GAP-OPA} drops sharply, while for \textbf{JAP-OPA}, the decrease is more gradual. When $\kappa$ exceeds $11$ dB, the infeasibility of \textbf{JAP-OPA} algorithm rises significantly for the given setup, leading to a sharp degradation in its performance as zero SE and failed sensing results are accounted for in infeasible scenarios.

\begin{figure}[t]
\centering
\includegraphics[width=0.44\textwidth]{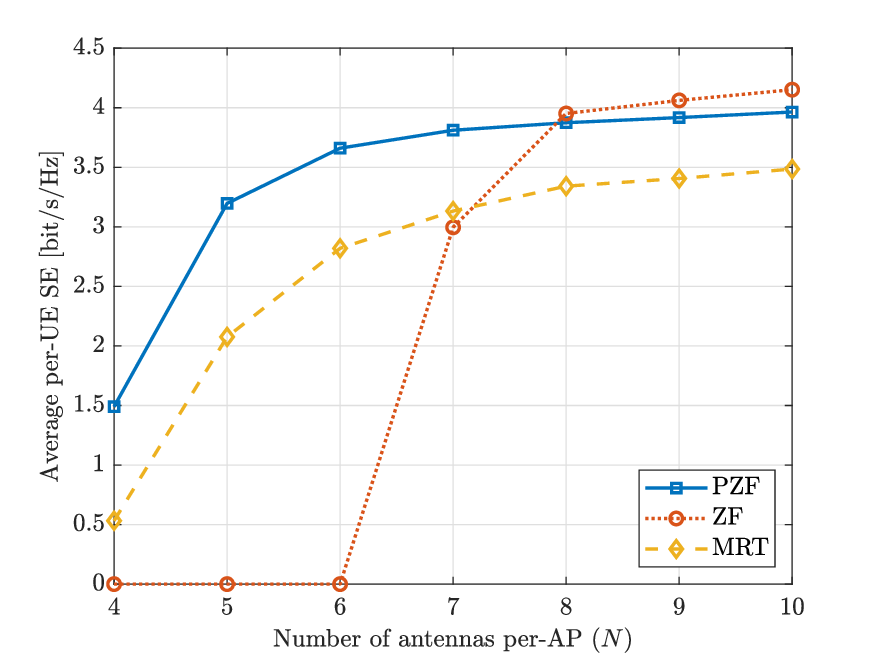}
\vspace{-0.5em}
\caption{\small  Average per-UE minimum SE versus $N$ ($\kappa=8$ dB, $M=60$, $K=6$, and $L=2$).}
\label{fig7}
\end{figure}

Figure~\ref{fig7} shows the average per-UE minimum SE for the PZF, ZF, and MRT precoding schemes as a function of the number of antennas at each AP, with $M=60$ and $K=6$. From this figure, we observe that when $N\le K$, ZF fails to support the UEs, while both MRT and PZF can support all UEs, with PZF providing the best performance. When the number of antennas per AP exceeds $8$, ZF begins to outperform the other schemes, as it gains sufficient degrees of freedom to effectively cancel interference among all UEs, while its multiplexing gain surpasses that of PZF. The reason ZF outperforms PZF is that, with the current threshold $\varrho$, some UEs still fall into the weak group for certain APs. In other words, a level of interference remains when the PZF scheme is employed. To overcome this limitation, the threshold can be adjusted to its optimal value. With a large number of antennas, this adjustment enables PZF to associate all users with the strong group. Accordingly, PZF and ZF would effectively converge, resulting in identical performance. 

\section{Conclusion}~\label{sec:conc}
In this paper, we proposed a distributed ISAC implementation underpinned by a CF-mMIMO architecture. We analyzed both the exact and asymptotic SE performance of the downlink communication system and provided exact results for the MASR metric of the sensing zones. The operation mode of the distributed APs and their transmit power coefficients were jointly optimized to maximize fairness among communication UEs, while ensuring a specific level of MASR for different sensing zones in the network. To reduce the computational complexity of the proposed optimized scheme, a low-complexity design was also developed, wherein the operation mode of the APs was determined through a greedy algorithm and then the transmit power of the APs was optimized via a bisection algorithm. The proposed AP mode selection and power control design was shown to provide significant performance over the benchmark systems with random/greedy AP mode selection for both communication and sensing operations. Exploring the processing of echo/reflected signals in a multi-static sensing scenario would be highly relevant. 

Future research directions could focus on AP mode selection for multi-static configurations, identifying the optimal receive APs for processing these reflected signals, and subsequently combining them at the CPU for enhanced localization and tracking purposes.  Another key characteristic of AP mode selection is the offered opportunity  to improve the security in cell-free ISAC systems. In fact, deliberately directing information beams towards sensing targets to improve the sensing performance can introduce a potential risk of information leakagee~\cite{Ren:JSAC:2024}. However, by employing our design, the likelihood of information leakage to the sensing (potentially suspicious) targets can be significantly reduced. A detailed investigation of secure ISAC systems, by considering AP mode selection, is indeed an interesting future research direction.

\vspace{-0.1em}
\appendix
\vspace{-0.6em}
\subsection{Useful Lemma}~\label{Apx:wPZF}
\vspace{-1.6em}
 \begin{Lemma}~\label{lemma:wPZF}
For the ZF beamforming vector $\qt^{\ZFcom}_{mk}$, defined in~\eqref{eq:Tzf}, we have
\vspace{-0.1em}
\begin{align}~\label{eq:outerprod_wPZF}
 \Ex\Big\{\tmkZFcom (\tmkZFcom)^\dag\Big\} &= \frac{\gammk}{N(N-|\Sm|)}\qI_N.
\end{align}
\end{Lemma}

\begin{proof}
Define 
$\qQ \!\triangleq \!\hat{\qG}_{\Sm} \big (\hat{\qG}_{\Sm} ^{\dag} \hat{\qG}_{\Sm}\big)^{ \!-1}$ 
 and $\qq_k=\qQ\qe_k$, which is  the $k$th column of $\qQ\in \mathbb{C}^{N\times \vert \Sm\vert}$. Then, we have
\vspace{-0.1em}
    \begin{align}
\gammk^2\Ex\Big\{\qq_k \qq_k^\dag\Big\}= \frac{\gammk^2}{\vert \Sm \vert  } \sum\nolimits^{\vert \Sm \vert}_{k'=1} \Ex\Big\{\qq_{k'} \qq_{k'}^\dag\Big\}.
    \end{align}
Accordingly, we have
\vspace{-0.1em}
\begin{align}
&\Ex\Big\{ \tmkZFcom (\tmkZFcom)^\dag \Big\}=\frac{\gammk^2}{\vert \Sm \vert  }  \Ex\big\{\qQ \qQ^\dag\big\}\nonumber\\
&  =\frac{\gammk^2}{\vert \Sm \vert  }   \Ex\Big\{ 
       \hat{\qG}_{\Sm}\! \left (\hat{\qG}_{\Sm}^{\dag}\! \hat{\qG}_{\Sm}\right)^{\!-1} \!
      \left (\!\hat{\qG}_{\Sm}^{\!H}\! \hat{\qG}_{\Sm}\!\right)^{\!-1}
      \hat{\qG}_{\Sm}^\dag\Big\}.
 \end{align}
Let $\qB=\Ex\big\{\hat{\qG}_{\Sm} \big (\hat{\qG}_{\Sm} ^{\dag} \hat{\qG}_{\Sm}\big) ^{-1} \big (\hat{\qG}_{\Sm} ^{\dag} \hat{\qG}_{\Sm}\big) ^{-1}     \hat{\qG}_{\Sm}^\dag\big\}$,   For any $N\times N$ unitary matrix $\BOmega$, we have       
      \vspace{-0.1em}
       \begin{align}~\label{eq:omBom}
    \BOmega  \qB\BOmega^{\dag}\! = \!\Ex\big\{  \BOmega
       \hat{\qG}_{\Sm} \big (\hat{\qG}_{\Sm} ^{\dag} \hat{\qG}_{\Sm}\big)^{\!-1} 
       \big (\hat{\qG}_{\Sm} ^{\dag} \hat{\qG}_{\Sm}\big) ^{\!\!-1}
      \hat{\qG}_{\Sm}^{\dag}\BOmega^{\dag}\big\}.
 \end{align}
Now, we define $  \bar{\qG}_{\Sm}=\BOmega \hat{\qG}_{\Sm}$. Thus, using~\eqref{eq:omBom}, we obtain
\vspace{-0.1em}
       \begin{align}
    \BOmega  \qB\BOmega^{\dag} \!=\!\! \Ex\big\{ 
       \bar{\qG}_{\Sm} \big(\bar{\qG}_{\Sm}^{\dag}\!\! \bar{\qG}_{\Sm}\big)^{-1} 
      \big (\bar{\qG}_{\Sm}^{\dag}\!\! \bar{\qG}_{\Sm}\big)^{\!-1}
      \bar{\qG}_{\Sm}^\dag\big\},
 \end{align}
where we have used
\vspace{-0.1em}
\begin{align}
\bar{\qG}^{\dag}_{\Sm} \bar{\qG}_{\Sm}&=\hat{\qG}^{\dag}_{\Sm}  \BOmega^{\dag}  \BOmega \hat{\qG}_{\Sm}  
       = \hat{\qG}^{\dag}_{\Sm} \hat{\qG}_{\Sm}. 
 \end{align}

Since $\bar{\qG}_{\Sm}=   \BOmega \hat{\qG}_{\Sm}$ is statistically identical to $\hat{\qG}_{\Sm}$, we have $\BOmega  \qB\BOmega^{\dag} = \qB$, for any unitary matrix  $\BOmega$.
By using an eigenvalue decomposition, $\qB$ can be expressed as $\qB=\qW \qD_{\lambda} \qW^{\dag}$, where $\qW$ is a unitary matrix and $\qD_{\lambda}$ is a diagonal matrix. Then, $\BOmega  \qB\BOmega^{\dag} = \qB$ is equivalent to
\begin{align}\label{eq:star}
\BOmega \qW \qD_{\lambda} \BOmega^{\dag} \qW^{\dag} = \qU \qD_{\lambda} \qU^{\dag}=\qB.
\end{align}
Since~\eqref{eq:star} is true for any unitary $\qU$, $\qB$ must be a scaled identity matrix. This implies that $\qB=c_1 \qI_N$, where $c_1$ is a constant, given by
\vspace{-0.1em}
\begin{align}
c_1&= \frac{1}{N} \Ex\Big\{ \trac \Big(
       \hat{\qG}_{\Sm} \left (\hat{\qG}_{\Sm} ^{\dag} \hat{\qG}_{\Sm}\right)^{ \!-1} 
      \left (\hat{\qG}_{\Sm}^{\dag} \hat{\qG}_{\Sm}\right)^{ \!-1}
      \hat{\qG}_{\Sm}^\dag\Big)\Big\} \nonumber\\
      &
      = \frac{1}{N} \Ex\Big\{ \trac (
      \hat{\qG}_{\Sm} ^{\dag} \hat{\qG}_{\Sm}) ^{-1} \Big\}
      \nonumber\\
      &
      =\frac{\vert \Sm\vert}{ N(N-\vert \Sm \vert)\gammk}.
\end{align}

\vspace{-0.4em}
\subsection{Proof of Proposition~\ref{Prop:SE:PPZF}}
\label{Prop:SE:PPZF:proof}
We first present the desired signal term as
\vspace{0.4em}
\begin{align}~\label{eq:DS:ap}
\mathrm{DS}_k  
 &= 
\Ex\Big\{ \sum\nolimits_{m\in\ZK}  \sqrt{a_m\rho \etamkc}  (\hat{\qg}_{m k} +\tilde{\qg}_{m k})^{\dag} \qt^{\ZFcom}_{mk} \nonumber\\
&\hspace{1em}+\sum\nolimits_{m\in\MK}  \sqrt{a_m\rho \etamkc}  (\hat{\qg}_{m k} +\tilde{\qg}_{m k})^{\dag} \qt^{\MRcom}_{mk}  \Big\}
\nonumber\\
 &
=\sum\nolimits_{m\in\ZK}\!\!\!  \sqrt{a_m\rho\etamkc}\gammk\!\!+\!\!N\!\sum\nolimits_{m\in\MK} \!\! \sqrt{a_m\rho\etamkc} \gammk,
\end{align}
where we have used the fact that $ \qt^{\ZFcom}_{mk}$, $ \qt^{\MRcom}_{mk}$ and $ \tilde{\qg}_{m k}$ are zero mean and independent. 

We can now proceed with the following derivation 
\vspace{0.4em}
\begin{align}
 &\Ex\big\{ \big\vert  \mathrm{BU}_k  \big\vert^2\big\}= 
 \Ex\big\{ \big\vert\sum\nolimits_{m\in\ZK}  \sqrt{a_m\rho \etamkc} \gtmk \qt^{\ZFcom}_{mk} \nonumber\\
&\hspace{1em} 
 +\sum\nolimits_{m\in\MK}  \sqrt{a_m\rho \etamkc} 
 \gtmk \qt^{\MRcom}_{mk} 
 \big\vert^2\big\}-\mathrm{DS}_k.
\end{align}

Since the variance of the sum of independent RVs is equal to the sum of the variances, we have
\vspace{0.4em}
\begin{align} ~\label{eq:BU:1}
\vartheta_k&=\Ex\Big\{ \Big\vert\sum\nolimits_{m\in\ZK}  \sqrt{a_m\rho \etamkc}  \gtmk \qt^{\ZFcom}_{mk} \nonumber\\
&\hspace{2em} +\sum\nolimits_{m\in\mathcal{M}_k}  \sqrt{a_m\rho \etamkc} \gtmk \qt^{\MRcom}_{mk} \Big\vert^2\Big\} \nonumber\\
&
=\vartheta_{Z_k}+\vartheta_{M_k}+2\Big( \sum\nolimits_{m\in\ZK}  \sqrt{a_m\rho \etamkc}\gammk   \Big) \nonumber\\
&\hspace{2em} \times \Big( N\sum\nolimits_{{m\in\mathcal{M}_k}} \sqrt{a_m\rho \etamkc}\gammk   \Big),
\end{align}
where 
\vspace{0.4em}
\begin{align}
\vartheta_{Z_k} &= \Ex\Big\{ \big\vert{{\sum\nolimits_{m\in\ZK}  \sqrt{a_m\rho \etamkc} \gtmk \qt^{\ZFcom}_{mk} }} \big\vert^2\Big\}\nonumber\\
\vartheta_{M_k}&=  \Ex\Big\{ \big\vert{{\sum\nolimits_{m\in\MK}  \sqrt{a_m\rho \etamkc} \gtmk \qt^{\MRcom}_{mk} }} \big\vert^2\Big\}.
\end{align}

We first focus on $\vartheta_{Z_k}$, which can be computed as
\vspace{0.4em}
\begin{align}~\label{eq:BU2:ap}
\vartheta_{Z_k} 
&= \Ex\Big\{ \Big\vert\sum\nolimits_{m\in\ZK}  \sqrt{a_m\rho \etamkc} (\hat{\qg}_{m k} +\tilde{\qg}_{m k})^{\dag} \qt^{\ZFcom}_{mk}  \Big\vert^2\Big\}\nonumber\\
&
=\sum\nolimits_{m\in\ZK} a_m \rho \etamkc\nonumber\\
&\hspace{2em} \times 
\Ex\Big\{  (\qt^{\ZFcom}_{mk})^{\dag}
\Ex\Big\{    \tilde{\qg}^{\dag}_{m k}  (\tilde{\qg}^{\dag}_{m k})^{\dag} \Big\} \qt^{\ZFcom}_{mk}  \Big\} \nonumber\\
&
= \sum\nolimits_{m\in\ZK}\!\!\!  \sqrt{a_m\rho\etamkc}\gammk\sum\nolimits_{m\in\ZK} \!\!\!\!a_m \rho \etamkc(\beta_{mk}\!\!-\! \gammk) 
\nonumber\\
&\hspace{2em} \times 
\gammk^2\Ex\Big\{ \!  \big[(\qG_{\Sm} ^\dag \qG_{\Sm}\!)^{\!-1} \big]_{kk}   \! \Big\} \nonumber\\
&
= \sum\nolimits_{m\in\ZK}\!\!\!  \sqrt{a_m\rho\etamkc}\gammk\nonumber\\
&\hspace{1em}+\sum\nolimits_{m\in\ZK} a_m \rho \etamkc(\beta_{mk}-\gammk) \frac{\gammk}{N-|\Sm|}.
\end{align}
Then, the second term of \eqref{eq:BU:1},  $\vartheta_M$, can be derived as
\vspace{0.4em}
\begin{align}~\label{eq:BU1:ap}
\vartheta_{M_k} 
 &
 =\sum\nolimits_{m\in\MK} \!\!\!a_m \rho \etamkc
 \text{var}\Big((\|\hat{\qg}_{m k}\|^2 +\tilde{\qg}_{m k}^{\dag}\hat{\qg}_{m k}\Big)\nonumber\\
 &+
 \bigg\vert\Ex\Big\{\sum\nolimits_{m\in\MK} a_m \sqrt{\rho\etamkc} 
 (\hat{\qg}_{m k} +\tilde{\qg}_{m k})^{\dag}\qt^{\MRcom}_{mk}\Big\}\bigg\vert^2 \nonumber\\
&
 =\sum\nolimits_{m\in\MK} a_m \rho\etamkc N \gammk\betmk\nonumber\\
 &\hspace{4em} +\bigg(N\sum\nolimits_{m\in\MK} a_m \sqrt{\rho\etamkc}\gammk\bigg)^2.
 \end{align}
 As a result, plugging \eqref{eq:BU2:ap} and \eqref{eq:BU1:ap} into \eqref{eq:BU:apf}, we have
 \vspace{0.4em}
\begin{align}~\label{eq:BU:apf}
 &\Ex\Big\{ \big\vert  \mathrm{BU}_k  \big\vert^2\Big\}
= \rho\sum\nolimits_{m\in\ZK} a_m \etamkc(\betmk-\gammk) \frac{\gamma_{mk}}{N-|\Sm|}
\nonumber\\
&\hspace{6em}+\rho\sum\nolimits_{m\in\MK} N a_m  \etamkc  \gammk\betmk.
\end{align}

The same steps can be followed to compute $  \mathrm{IUI}_{kk'}$:
\begin{align}~\label{eq:IR:ap}
\mathrm{IUI}_{kk'}=& \rho\sum\nolimits_{m\in\ZK} a_m \etamkpc(\betmk-\gammk)
\frac{\gamma_{mk'}}{N-|\Sm|} 
\nonumber\\
&
+N\sum\nolimits_{m\in\MK} a_m  \etamkpc \gammk\beta_{mk'}.
\end{align} 
Moreover, by using~\eqref{eq:IRk}, we get
\begin{align}~\label{eq:IR:ap2}
  \Ex\Big\{ \big\vert \IR_{k} \big\vert^2\Big\} 
   &= \rho\sum\nolimits_{m\in\M}\sum\nolimits _{l\in\LL}\rho\etamls(1-a_m) \nonumber\\
  &\times\Ex\Big\{\vert\qg_{mk}^\dag \qa_N(\theta_{t,ml})\vert^2\Big\}.
\end{align}
Letting $\qA=\qa_N(\theta_{t,ml})\qa^\dag_N(\theta_{t,ml})$, we have 
\begin{align}~\label{eq:IR:ap2}
  &\Ex\Big\{ \big\vert \IR_{k} \big\vert^2\Big\} 
  =\!\rho\sum\nolimits_{m\in\M} \sum\nolimits _{l\in\LL}\etamls(1-a_m) \Ex\Big\{\vert\qg_{mk}^T \qA\vert^2\Big\}
  \nonumber\\
   &=\!\rho\sum\nolimits_{m\in\M}\sum\nolimits _{l\in\LL}\etamls(1-a_m) \Ex\Big\{\qg_{mk}^T \qA\qg_{mk}^*\Big\}
    \nonumber\\
&=\!\rho\sum\nolimits_{m\in\M}\!\!\sum\nolimits _{l\in\LL}\!\!\etamls(1\!-\!a_m) \betmk\trac(\qA).
  \end{align}

We notice that $\trac(\qA)=N$.  To this end, by substituting~\eqref{eq:DS:ap},~\eqref{eq:BU:apf},~\eqref{eq:IR:ap}, and~\eqref{eq:IR:ap2} into~\eqref{eq:dLSE}, the desired result in~\eqref{eq:dLSNIRf} is obtained.
\end{proof}


\balance
\bibliographystyle{IEEEtran}
\bibliography{IEEEabrv,ref_bib}

\begin{thebibliography}{10}
\providecommand{\url}[1]{#1}
\csname url@samestyle\endcsname
\providecommand{\newblock}{\relax}
\providecommand{\bibinfo}[2]{#2}
\providecommand{\BIBentrySTDinterwordspacing}{\spaceskip=0pt\relax}
\providecommand{\BIBentryALTinterwordstretchfactor}{4}
\providecommand{\BIBentryALTinterwordspacing}{\spaceskip=\fontdimen2\font plus
\BIBentryALTinterwordstretchfactor\fontdimen3\font minus
  \fontdimen4\font\relax}
\providecommand{\BIBforeignlanguage}[2]{{%
\expandafter\ifx\csname l@#1\endcsname\relax
\typeout{** WARNING: IEEEtran.bst: No hyphenation pattern has been}%
\typeout{** loaded for the language `#1'. Using the pattern for}%
\typeout{** the default language instead.}%
\else
\language=\csname l@#1\endcsname
\fi
#2}}
\providecommand{\BIBdecl}{\relax}
\BIBdecl

\bibitem{elfiatoure2023cell}
M.~Elfiatoure, M.~Mohammadi, H.~Q. Ngo, and M.~Matthaiou, ``Cell-free massive
  {MIMO} for {ISAC}: Access point operation mode selection and power control,''
  in \emph{Proc. IEEE GLOBECOM}, Dec. 2023.

\bibitem{liu2022integrated:dul}
F.~Liu, Y.~Cui, C.~Masouros, J.~Xu, T.~X. Han, Y.~C. Eldar, and S.~Buzzi,
  ``Integrated sensing and communications: {Toward} dual-functional wireless
  networks for {6G} and beyond,'' \emph{{IEEE} J. Sel. Areas Commun.}, vol.~40,
  no.~6, pp. 1728--1767, Jun. 2022.

\bibitem{elfiatoure2023coexistence}
M.~Elfiatoure, M.~Mohammadi, H.~Q. Ngo, P.~J. Smith, and M.~Matthaiou,
  ``Protecting massive {MIMO}-radar coexistence: Precoding design and power
  control,'' \emph{IEEE Open J. Commun. Society}, vol.~5, pp. 276--293, Jan.
  2024.

\bibitem{Elfiatoure:JCIN:2023}
M.~Elfiatoure, H.~Q. Ngo, and M.~Matthaiou, ``Coexistence between massive
  {MIMO} and radar communications: Performance analysis,'' \emph{J. Commun.
  Inf. Networks}, vol.~8, no.~1, pp. 37--47, Mar. 2023.

\bibitem{liu2020joint}
F.~Liu, C.~Masouros, A.~Petropulu, H.~Griffiths, and L.~Hanzo, ``Joint radar
  and communication design: Applications, state-of-the-art, and the road
  ahead,'' \emph{{IEEE} Trans. Commun.}, vol.~68, no.~6, pp. 3834--3862, Jun.
  2020.

\bibitem{liu2018dual}
F.~Liu, L.~Zhou, C.~Masouros, A.~Lit, W.~Luo, and A.~Petropulu,
  ``Dual-functional cellular and radar transmission: Beyond coexistence,'' in
  \emph{Proc. IEEE SPAWC}, Jun. 2018, pp. 1--5.

\bibitem{liu2022seventy}
F.~Liu, L.~Zheng, Y.~Cui, C.~Masouros, A.~P. Petropulu, H.~Griffiths, and Y.~C.
  Eldar, ``Seventy years of radar and communications: The road from separation
  to integration,'' \emph{{IEEE} IEEE Signal Process. Mag.}, vol.~40, no.~5,
  pp. 106--121, Jul. 2023.

\bibitem{liu2018mu}
F.~Liu, C.~Masouros, A.~Li, H.~Sun, and L.~Hanzo, ``{MU-MIMO} communications
  with {MIMO} radar: {From} co-existence to joint transmission,'' \emph{{IEEE}
  Trans. Wireless Commun.}, vol.~17, no.~4, pp. 2755--2770, Apr. 2018.

\bibitem{ali2020leveraging}
A.~Ali, N.~Gonzalez-Prelcic, R.~W. Heath, Jr., and A.~Ghosh, ``Leveraging
  sensing at the infrastructure for mm{W}ave communication,'' \emph{IEEE
  Commun. Mag.}, vol.~58, no.~7, pp. 84--89, Jul. 2020.

\bibitem{zhang2020perceptive}
A.~Zhang, M.~L. Rahman, X.~Huang, Y.~J. Guo, S.~Chen, and R.~W. Heath, Jr.,
  ``Perceptive mobile networks: Cellular networks with radio vision via joint
  communication and radar sensing,'' \emph{IEEE Veh. Technol. Mag.}, vol.~16,
  no.~2, pp. 20--30, Jun. 2020.

\bibitem{temiz2020dual}
M.~Temiz, E.~Alsusa, and M.~W. Baidas, ``A dual-functional massive {MIMO OFDM}
  communication and radar transmitter architecture,'' \emph{IEEE Trans. Veh.
  Technol.}, vol.~69, no.~12, pp. 14\,974--14\,988, Dec. 2020.

\bibitem{temiz2021optimized}
------, ``Optimized precoders for massive {MIMO OFDM} dual radar-communication
  systems,'' \emph{IEEE Trans. Commun.}, vol.~69, no.~7, pp. 4781--4794, Jul.
  2021.

\bibitem{temiz2021dual}
------, ``A dual-function massive {MIMO} uplink {OFDM} communication and radar
  architecture,'' \emph{IEEE Trans. Cognit. Commun. Netw.}, vol.~8, no.~2, pp.
  750--762, Jun. 2021.

\bibitem{Mohammadi:PROC.2024}
M.~Mohammadi, Z.~Mobini, H.~Quoc~Ngo, and M.~Matthaiou, ``Next-generation
  multiple access with cell-free massive {MIMO},'' \emph{Proc. {IEEE}}, vol.
  112, no.~9, pp. 1372--1420, Sept. 2024.

\bibitem{ngo2024ultra}
H.~Q. Ngo, G.~Interdonato, E.~G. Larsson, G.~Caire, and J.~G. Andrews,
  ``Ultradense cell-free massive {MIMO} for {6G}: Technical overview and open
  questions,'' \emph{Proc. {IEEE}}, vol. 112, no.~7, pp. 805--831, Jul. 2024.

\bibitem{Matthaiou:COMMag:2021}
\emph{\emph{M. Matthaiou}} \emph{et al.}, ``The road to {6G: Ten} physical
  layer challenges for communications engineers,'' \emph{{IEEE} Commun. Mag.},
  vol.~59, no.~1, pp. 64--69, Jan. 2021.

\bibitem{Hien:cellfree}
H.~Q. Ngo, A.~Ashikhmin, H.~Yang, E.~G. Larsson, and T.~L. Marzetta,
  ``Cell-free massive {MIMO} versus small cells,'' \emph{{IEEE} Trans. Wireless
  Commun.}, vol.~16, no.~3, pp. 1834--1850, Mar. 2017.

\bibitem{zhang2020prospective}
J.~Zhang, E.~Bj{\"o}rnson, M.~Matthaiou, D.~W.~K. Ng, H.~Yang, and D.~J. Love,
  ``Prospective multiple antenna technologies for beyond {5G},'' \emph{{IEEE}
  J. Sel. Areas Commun.}, vol.~38, no.~8, pp. 1637--1660, Aug. 2020.

\bibitem{Mohammadi:TWC:2024}
M.~Mohammadi, Z.~Mobini, H.~Q. Ngo, and M.~Matthaiou, ``Ten years of research
  advances in full-duplex massive {MIMO},'' \emph{{IEEE} Trans. Commun.}, pp.
  1--1, 2024.

\bibitem{zeng2023integrated}
F.~Zeng, J.~Yu, J.~Li, F.~Liu, D.~Wang, and X.~You, ``Integrated sensing and
  communication for network-assisted full-duplex cell-free distributed massive
  {MIMO} systems,'' \emph{arXiv preprint arXiv:2311.05101}, 2023.

\bibitem{Behdad:GC:2022}
Z.~Behdad, O.~T. Demir, K.~W. Sung, E.~Björnson, and C.~Cavdar, ``Multi-static
  target detection and power allocation for integrated sensing and
  communication in cell-free massive {MIMO},'' \emph{{IEEE} Trans. Wireless
  Commun.}, vol.~23, no.~9, pp. 11\,580--11\,596, Sept. 2024.

\bibitem{demirhan2023cell}
U.~Demirhan and A.~Alkhateeb, ``Cell-free {ISAC MIMO} systems: Joint sensing
  and communication beamforming,'' \emph{{IEEE} Trans. Commun.}, pp. 1--1,
  2024.

\bibitem{da2023multi}
I.~W. da~Silva, D.~P. Osorio, and M.~Juntti, ``Multi-static {ISAC} in cell-free
  massive {MIMO}: Precoder design and privacy assessment,'' in \emph{Proc. IEEE
  GLOBECOM}, Dec. 2023, pp. 461--466.

\bibitem{Mao:TWC:2024}
W.~Mao, Y.~Lu, C.-Y. Chi, B.~Ai, Z.~Zhong, and Z.~Ding, ``Communication-sensing
  region for cell-free massive {MIMO ISAC} systems,'' \emph{{IEEE} Trans.
  Wireless Commun.}, vol.~23, no.~9, pp. 12\,396--12\,411, Sept. 2024.

\bibitem{Stoica:TSP:2007}
P.~Stoica, J.~Li, and Y.~Xie, ``On probing signal design for {MIMO} radar,''
  \emph{{IEEE} Trans. Signal Process.}, vol.~55, no.~8, pp. 4151--4161, Aug.
  2007.

\bibitem{Liu:TSP:2020}
X.~Liu, T.~Huang, N.~Shlezinger, Y.~Liu, J.~Zhou, and Y.~C. Eldar, ``Joint
  transmit beamforming for multiuser {MIMO} communications and {MIMO} radar,''
  \emph{{IEEE} Trans. Signal Process.}, vol.~68, pp. 3929--3944, Jun. 2020.

\bibitem{Dong:TGCN:2023}
F.~Dong, W.~Wang, X.~Li, F.~Liu, S.~Chen, and L.~Hanzo, ``Joint beamforming
  design for dual-functional {MIMO} radar and communication systems
  guaranteeing physical layer security,'' \emph{IEEE Trans. Green Commun.
  Netw.}, vol.~7, no.~1, pp. 537--549, Mar. 2023.

\bibitem{Zhang:TWC:2024}
S.~Zhang, W.~Hao, G.~Sun, C.~Huang, Z.~Zhu, X.~Li, and C.~Yuen, ``Joint
  beamforming optimization for active {STAR-RIS}-assisted {ISAC} systems,''
  \emph{{IEEE} Trans. Wireless Commun.}, pp. 1--1, 2024.

\bibitem{Chen:TCOM:2024}
J.~Chen, K.~Wu, J.~Niu, Y.~Li, P.~Xu, and J.~A. Zhang, ``Spectral and energy
  efficient waveform design for {RIS}-assisted {ISAC},'' \emph{{IEEE} Trans.
  Commun.}, pp. 1--1, 2024.

\bibitem{Interdonato:TCOM:2020}
G.~Interdonato, M.~Karlsson, E.~Björnson, and E.~G. Larsson, ``Local partial
  zero-forcing precoding for cell-free massive {MIMO},'' \emph{{IEEE} Trans.
  Wireless Commun.}, vol.~19, no.~7, pp. 4758--4774, Apr. 2020.

\bibitem{Zhao:CLET:2022}
N.~Zhao, Y.~Wang, Z.~Zhang, Q.~Chang, and Y.~Shen, ``Joint transmit and receive
  beamforming design for integrated sensing and communication,'' \emph{{IEEE}
  Commun. Lett.}, vol.~26, no.~3, pp. 662--666, Mar. 2022.

\bibitem{Ngo:TGCN:2018}
H.~Q. Ngo, L.-N. Tran, T.~Q. Duong, M.~Matthaiou, and E.~G. Larsson, ``On the
  total energy efficiency of cell-free massive {MIMO},'' \emph{IEEE Trans.
  Green Commun. Netw.}, vol.~2, no.~1, pp. 25--39, Mar. 2018.

\bibitem{emil20TWC}
E.~{Bj\"{o}rnson} and L.~{Sanguinetti}, ``Making cell-free massive {MIMO}
  competitive with {MMSE} processing and centralized implementation,''
  \emph{IEEE Trans. Wireless Commun.}, vol.~19, no.~1, pp. 77--90, Jan. 2020.

\bibitem{Mohammadi:JSAC:2023}
M.~Mohammadi, T.~T. Vu, H.~Q. Ngo, and M.~Matthaiou, ``Network-assisted
  full-duplex cell-free massive {MIMO}: Spectral and energy efficiencies,''
  \emph{{IEEE} J. Sel. Areas Commun.}, vol.~41, no.~9, pp. 2833--2851, Sept.
  2023.

\bibitem{Liu:TSP:2018}
F.~Liu, L.~Zhou, C.~Masouros, A.~Li, W.~Luo, and A.~Petropulu, ``Toward
  dual-functional radar-communication systems: Optimal waveform design,''
  \emph{{IEEE} Trans. Signal Process.}, vol.~66, no.~16, pp. 4264--4279, Aug.
  2018.

\bibitem{Hua:TWC:2024}
M.~Hua, Q.~Wu, W.~Chen, O.~A. Dobre, and A.~L. Swindlehurst, ``Secure
  intelligent reflecting surface-aided integrated sensing and communication,''
  \emph{{IEEE} Trans. Wireless Commun.}, vol.~23, no.~1, pp. 575--591, Jan.
  2024.

\bibitem{Ren:TCOM:2023}
Z.~Ren, L.~Qiu, J.~Xu, and D.~W.~K. Ng, ``Robust transmit beamforming for
  secure integrated sensing and communication,'' \emph{{IEEE} Trans. Commun.},
  vol.~71, no.~9, pp. 5549--5564, Sept. 2023.

\bibitem{Zhenyao:JSAC:2023}
Z.~He, W.~Xu, H.~Shen, D.~W.~K. Ng, Y.~C. Eldar, and X.~You, ``Full-duplex
  communication for {ISAC}: Joint beamforming and power optimization,''
  \emph{{IEEE} J. Sel. Areas Commun.}, vol.~41, no.~9, pp. 2920--2936, Sept.
  2023.

\bibitem{Liao:TWC:2024}
B.~Liao, H.~Q. Ngo, M.~Matthaiou, and P.~J. Smith, ``Power allocation for
  massive {MIMO-ISAC} systems,'' \emph{{IEEE} Trans. Wireless Commun.},
  vol.~23, no.~10, pp. 14\,232--14\,248, Oct. 2024.

\bibitem{cramer2004random}
H.~Cram{\'e}r, \emph{{R}andom {V}ariables and {P}robability
  {D}istributions}.\hskip 1em plus 0.5em minus 0.4em\relax Cambridge University
  Press, 2004, no.~36.

\bibitem{Ngo:TCOM:2013}
H.~Q. Ngo, E.~G. Larsson, and T.~L. Marzetta, ``Energy and spectral efficiency
  of very large multiuser {MIMO} systems,'' \emph{{IEEE} Trans. Commun.},
  vol.~61, no.~4, pp. 1436--1449, Apr. 2013.

\bibitem{Che:TWC:2014}
E.~Che, H.~D. Tuan, and H.~H. Nguyen, ``Joint optimization of cooperative
  beamforming and relay assignment in multi-user wireless relay networks,''
  \emph{{IEEE} Trans. Wireless Commun.}, vol.~13, no.~10, pp. 5481--5495, Oct.
  2014.

\bibitem{cvx}
M.~Grant and S.~Boyd, ``{CVX: Matlab} software for disciplined convex
  programming, version 2.1, [online]. available:http: //cvxr.com/cvx, 2014.''
  2014.

\bibitem{vu18TCOM}
T.~T. {Vu}, D.~T. {Ngo}, M.~N. {Dao}, S.~{Durrani}, and R.~H. {Middleton},
  ``Spectral and energy efficiency maximization for content-centric {C-RANs}
  with edge caching,'' \emph{IEEE Trans. Commun.}, vol.~66, no.~12, pp.
  6628--6642, Dec. 2018.

\bibitem{tam16TWC}
H.~H.~M. Tam, H.~D. Tuan, D.~T. Ngo, T.~Q. Duong, and H.~V. Poor, ``Joint load
  balancing and interference management for small-cell heterogeneous networks
  with limited backhaul capacity,'' \emph{IEEE Trans. Wireless Commun.},
  vol.~16, no.~2, pp. 872--884, Feb. 2017.

\bibitem{Ren:JSAC:2024}
Z.~Ren, J.~Xu, L.~Qiu, and D.~Wing Kwan~Ng, ``Secure cell-free integrated
  sensing and communication in the presence of information and sensing
  eavesdroppers,'' \emph{{IEEE} J. Sel. Areas Commun.}, vol.~42, no.~11, pp.
  3217--3231, Nov. 2024.

\end{thebibliography}

\end{document}